\documentclass[11pt]{article}
\usepackage{hyperref}

\usepackage[T1]{fontenc}
\usepackage{cite}
\usepackage{texdraw}
\usepackage{fullpage}
\usepackage{algorithm}
\usepackage{bm}
\usepackage[noend]{algorithmic}
\usepackage{cite}
\usepackage{multirow}
\usepackage{nicefrac}

\usepackage{microtype}
\usepackage{subfigure}
\usepackage{booktabs} %
\usepackage{tikz}
\usetikzlibrary{positioning}

\usepackage[utf8]{inputenc}
\usepackage{xspace}
\usepackage{multicol}
\usepackage{graphicx}
\usepackage{enumitem}
\usepackage{amsmath,amssymb}
\usepackage{nicefrac}
\usepackage{float}
\usepackage{array}
\usepackage{tabu}
\usepackage{epsfig}
\usepackage{epstopdf}
\usepackage{amsthm,amsfonts,amsmath}
\usepackage{color}
\usepackage{balance}
\usepackage{url}
\usepackage{accents}
\usepackage{cleveref}
\usepackage[a4paper, margin=1in]{geometry} %

\DeclareMathOperator{\geom}{Geom}

\DeclareMathOperator*{\E}{E}
\DeclareMathOperator*{\Var}{Var}

\DeclareMathOperator*{\Bin}{Bin}
\DeclareMathOperator*{\Poisson}{Poisson}
\DeclareMathOperator*{\CM}{CM}
\DeclareMathOperator*{\irregularcuckoohashing}{ICH}

\DeclareMathOperator*{\dTV}{\textit{d}_{TV}}

\DeclareMathOperator{\arctanh}{arctanh}

\newcommand{\narrowldots}{\mathinner{\ldotp\!\ldotp\!\ldotp}}
\newcommand{\maxdeg}{M}

\allowdisplaybreaks

\usepackage[appendix=append,bibliography=common]{apxproof}

\newtheoremrep{theorem}{Theorem}
\newtheoremrep{lemma}[theorem]{Lemma}
\newtheoremrep{corollary}[theorem]{Corollary}
\newtheoremrep{definition}[theorem]{Definition}
\newtheoremrep{observation}[theorem]{Observation}
\newtheoremrep{fact}[theorem]{Fact}
\usepackage{colortbl}

\title{A New Impossibility Result for Online Bipartite Matching Problems%
}

\author{Flavio Chierichetti\thanks{Supported in part by a  Google Focused Research Award,   by
BiCi -- Bertinoro international Center for informatics, and by the PRIN project 20229BCXNW funded by the European Union - Next Generation EU, Mission 4 Component 1,
CUP B53D23012910006.}\\
Sapienza University of Rome\\
\texttt{\large flavio@di.uniroma1.it}
\and
Mirko Giacchini\footnotemark[1]\\
Sapienza University of Rome\\
\texttt{\large giacchini@di.uniroma1.it}
\and
Alessandro Panconesi\footnotemark[1]\\
Sapienza University of Rome\\
\texttt{\large ale@di.uniroma1.it}
\and
Andrea Vattani\\
Reddit\\
\texttt{\large andrea.vattani@reddit.com}
}

\date{}

\begin{document}

\maketitle

\begin{abstract}
Online Bipartite Matching with random user arrival is a fundamental problem in the online advertisement ecosystem. Over the last 30 years, many algorithms and impossibility results have been developed for this problem. In particular, the latest impossibility result was established by Manshadi, Oveis Gharan and Saberi~\cite{mos11} in 2011. Since then, several algorithms have been published in an effort to narrow the gap between the upper and the lower bounds on the competitive ratio.

\smallskip

In this paper we show that no algorithm can achieve a competitive ratio better than $1- \frac e{e^e} = 0.82062\ldots$, improving upon the $0.823$ upper bound presented in~\cite{mos11}.
Our construction is simple to state, accompanied by a fully analytic proof, and yields a competitive ratio bound  intriguingly similar to $1 - \frac1e$, the optimal competitive ratio for the fully adversarial Online Bipartite Matching problem. 

\smallskip

Although the tightness of our upper bound remains an open question, we show that our construction is extremal in a natural class of instances.
\end{abstract}

\section{Introduction}

Online Bipartite Matching problems are central to the online advertising ecosystem \cite{mos11,my11,kmt11,mehta13} and have numerous other applications, including efficient packet switching and routing \cite{az06,ap02}, crowdsourcing tasks \cite{tsdwc16},  market clearing problems \cite{bsz06} and ride-sharing platform optimization \cite{dssx21}.  
In this paper, we establish a new upper bound of $1 - e^{1-e}$
on the competitive ratio for several well-studied variants of the problem. %
This result marks the first improvement of the upper bound for this classic problem in more than a decade.

To explain our result in context, we begin by recalling the problem. We are given a bipartite graph where one side of the bipartition consists of a known set of ``ads'' (or ``advertisers''), while the other side consists of ``users'' (or ``ad slots''), whose neighbor sets are initially unknown. Users arrive sequentially, %
and the generic user's neighbors are revealed upon arrival. When a user arrives, the online algorithm can match, irrevocably, the user to one of its available neighbors; the algorithm's aim is to maximize the size of the matching at the end of the sequence.
The competitive ratio of an algorithm, in this case, is defined as the expected size of the matching produced by the algorithm divided by the expected size of the maximum possible matching. (We consider expectations since the algorithm can be randomized, and the input graph can be sampled from a distribution).

A strong motivation for studying this problem comes from online advertising, where edges represent ads that a user is interested in. %
Budget constraints limit how often each ad can be shown; without loss of generality, we assume that
each ad can be matched to at most one user (ads with budgets allowing for multiple impressions can be reproduced multiple times in the graph).
In this setting, a matching corresponds to a set of ads that will be %
shown to users, and %
the competitive ratio indicates how effectively the algorithm maximizes revenue.

This problem was introduced by Karp, Vazirani and Vazirani in a classic paper \cite{kvv90}. In their setup, the graph and the ordering of the users were chosen adversarially, and they presented a simple algorithm called RANKING that achieves a competitive ratio of \( 1 - e^{-1} \) (see \cite{kvv90,gm08}), which was also proven to be tight \cite{kvv90}. %
The assumptions of adversarial graphs and adversarial user orderings have since been relaxed to develop algorithms with better competitive ratios. %

An important variant of the problem assumes that %
the user ordering is sampled uniformly at random, instead of being produced by an adversary \cite{gm08}.
Additionally, to move beyond adversarial graphs, researchers have explored probabilistic models of user behavior \cite{msvv05}. See~\cite{mehta13} for a comprehensive overview.

\smallskip

In particular, consider the power set \( 2^{A} \) of the set of ads $A$, and let \( P \) be a probability distribution over \( 2^A \). When a user arrives, a set \( S \) of ads is sampled according to \( P \), and the ads in \( S \) become the user's neighbors. In this setting, \( P \) may be known or unknown, and both cases have been studied.  This model of random bipartite graphs has been studied extensively and is motivated by modern machine learning-driven methods used to identify promising ads for users.

A particularly interesting special case of this last model, known  as the ``irregular cuckoo hashing'' model,\footnote{We recall that in (a standard variant of) cuckoo hashing \cite{pf04,dw07}, given an integer $k \ge 1$, each item $x$ is hashed to $k$ uniform-at-random slots of a hashtable; the item is stored in an empty slot, if one is available. In irregular cuckoo hashing~\cite{dgmmpr10}, one obtains $k$ by applying another hash function to $x$  --- this way, different items can be hashed to a different number of slots.} %
arises when a sample from \( P \) is obtained by first sampling the degree of a user from a random variable $D$ taking values over the non-negative integers, and then sampling uniformly-at-random with replacement a number of ads  equal to the sampled degree. %
This process generates a multiset of ads, but duplicate ads can be removed, leaving the remaining set as the user's neighbors. %

\smallskip

We can now state our result more precisely. We show that no online algorithm can achieve a competitive ratio greater than $1 - e^{1-e} = 0.82062\ldots$ across all the models discussed above.  %
We establish our bound in the irregular cuckoo hashing model and so, by definition, the bound extends to all the other models and is the strongest known bound for all those variants. %
The previous strongest upper bound, established in 2011 by Manshadi, Oveis Gharan, and Saberi~\cite{mos11}, was $0.823$. (Table~\ref{tab:bounds} summarizes the state of the art). While our improvement is small, it is comparable in magnitude to many previous algorithmics advancements for this problem. 

The %
upper bound of $0.823$ on the competitive ratio from~\cite{mos11} relies on the irregular cuckoo hashing constructions of Dietzfelbinger et al.~\cite{dgmmpr10} and was derived through numerical optimization, without yielding a closed-form expression.  
In contrast, our construction enables an exact analytical computation of its optimal competitive ratio. Furthermore, our %
construction uses simple rational probabilities, whereas the optimal probabilities in~\cite{mos11} remain unspecified and likely irrational. 

\begingroup
\renewcommand{\arraystretch}{1.5}
\begin{table}[t]
    \centering
    \small
    \begin{tabular}{|c|c|c|c|}
        \hline
        \multirow{2}{*}{\textbf{Setting}} & \multicolumn{3}{c|}{\textbf{Competitive Ratio}} \\
        \cline{2-4}
        & Best Algorithm & 
        Our UB & Previous Best UB  \\
        \hline
        \hline
        {\footnotesize Irregular Cuckoo Hashing}&\multirow{2}{*}{$0.716$ \cite{hsy22}}  & \multirow{4}{*}{$1 - \frac{e}{e^e} = 0.820\!\narrowldots$} & \multirow{4}{*}{$0.823$ \cite{mos11}}    \\
        \cline{1-1}
        {\footnotesize IID Users (Known Distribution)} & &  & \\
        \cline{1-2}
        {\footnotesize IID Users (Unknown Distribution)} & \multirow{2}{*}{$0.696$ \cite{my11}}&  &  \\
        \cline{1-1}
        {\footnotesize Adversarial Graph, UAR Ordering} &  &  &  \\
        \hline
        \hline
        {\footnotesize Adversarial Graph and Ordering} & $1-\frac1e = 0.632\!\narrowldots$ \cite{kvv90,gm08} & \cellcolor{gray!20} &$1-\frac1e = 0.632\!\narrowldots$ \cite{kvv90}   \\
        \hline
    \end{tabular}
    \caption{Bounds for various Online Bipartite Matching problem variants. Each row corresponds to a different variant, listed in order from the easiest to the hardest. %
    }
    \label{tab:bounds}
\end{table}
\endgroup

\smallskip

The random graphs used to prove our bound are based on a simple random variable \( D \) taking values over the non-negative integers:  
\[
\Pr\left[D=0\right] = \Pr\left[D=1\right] = 0,  \quad \Pr\left[D=d\right] = \frac{1}{d \cdot (d-1)}, \text{ for } d \ge 2,
\]
or, equivalently, $\Pr\left[D > d\right] = \nicefrac1d$ for each positive integer $d$.
As discussed earlier, when a user arrives, a value \( d \) is drawn from \( D \), and \( d \) ads are chosen uniformly at random (with replacement) as neighbors. The graph consists of \( n \) ads and \( n \) users. As we will clarify later, this distribution greedily ensures the ``maximum possible'' number of vertices of degree \( d \) for  \( d = 0,1,2,\ldots \), while maintaining a quasi-complete matching, i.e., one matching a $1-o(1)$ fraction of the users.

Intuitively, selecting user neighbors uniformly at random while keeping user degrees low impairs the performance of any online algorithm: over time, newly arriving users are likely to find all their randomly assigned neighbors already matched. However, establishing a strong upper bound also requires ensuring that the maximum matching remains large --- this is the main technical challenge in our analysis.

\smallskip

Whether our upper bound is tight for some or all of the four problem variants mentioned above remains an open question. 
Notably, the bound involves the Euler's number, in a manner analogous to the optimal competitive ratio of $1-\frac1e$ of the fully-adversarial Online Bipartite Matching problem.
While we cannot determine whether our upper bound can be improved, we show that it possesses a certain extremal property and that it is the strongest bound possible for a  natural class of instances.

\smallskip
\noindent
{\bf Paper Organization.} Section~\ref{sec:overview} gives an overview of our main construction, of its relation to previous upper bounds,  and of the techniques we used to study it. Section~\ref{sec:relatedwork} reviews related work. Section~\ref{sec:preliminaries} introduces key tools, techniques, and notation. In Section~\ref{sec:impossibility}, we present our construction and prove the main result: the \( 1 - {e}^{1-e} \) upper bound on the competitive ratio of Online Bipartite Matching Problems. Section~\ref{sec:tightness} establishes the extremality of our construction, showing that it greedily maximizes the number of low-degree users while ensuring a quasi-complete matching. Section~\ref{sec:uneven} extends our construction to arbitrary user-to-ad ratios, demonstrating that equibipartite graphs yield the strongest bound. We conclude with open questions in Section~\ref{sec:conclusion}.

\smallskip

All the proofs missing from the main body of this paper can be found in its Appendix.

\section{Overview}\label{sec:overview}
Our irregular cuckoo hashing instance contains $n$ users and $n$ ads, and assigns probability $\frac1{d \cdot (d-1)}$ to each user-degree $d \ge 2$ --- in other words, it posits that the probability that a user has degree strictly larger than the generic positive integer $d$ is exactly $\nicefrac1d$.
We obtained this distribution by greedily maximizing the number of users of degree $d =0,1,2,3,\ldots$ while guaranteeing that the random graph admits a quasi-complete matching from the user side --- a matching that pairs all users except for at most $o(n)$ of them. %
Observe that, in a cuckoo hashing instance, the smaller the user-degrees the less likely it is for an online algorithm to match users. The goal of our distribution is to minimize the number of online matches while guaranteeing that the graph admits a quasi-complete matching --- in other words, the distribution aims to minimize the numerator of the competitive ratio while %
making its denominator as large as possible. (We remark that most previous upper bounds were also based on  graphs admitting quasi-complete matchings; see, e.g., \cite{bk10,mos11,gm08,kvv90}.) %

\smallskip

{\bf Our Main Construction.} It is easy to see that, in order for a  cuckoo hashing instance containing $n$ ads, and no more than $n$ users, to be quasi-complete from the user side, the instance must contain at most $o(n)$ users of degree $0$ and $o(n)$ users of degree $1$. %
Furthermore, classical results can be used to prove that, if one aims to be quasi-complete from the user side, and all users have degree $2$, then the maximum number of  users is $c_2 \cdot n \pm o(n)$, for $c_2 = \frac12$.\footnote{The feasibility of $c_2 = \frac12$ can be obtained as a corollary of a classical result of Erd\"os and R\'enyi \cite{er60} which states that, for all small enough $\epsilon > 0$, $G\left(n,\left(\frac12-\epsilon\right)\cdot n\right)$ has, with probability $1-o(1)$, a $1-o(1)$ fraction of its vertices in connected components that are trees. By taking each vertex of $G\left(n,\left(\frac12-\epsilon\right) \cdot n\right)$ to be an ad, and each of its edges $\{u,v\}$ to be a user of degree $2$ connected to the ads $u$  and $v$, one deduces the claim from the observation that the edges of a forest can be oriented so that no vertex has in-degree larger than $1$.} %
Now, suppose that we must have $c_2 \cdot n \pm o(n)$ users of degree $2$. %
What is the largest number of users of degree $3$ that we can add while  still guaranteeing that there exists a matching covering a $1-o(1)$ fraction of the users? We  prove that the answer %
is ``$c_3 \cdot  n \pm o(n)$, for $c_3 = \frac16$''. Next, if we are bound to have $c_2 \cdot n \pm o(n)$ users of degree $2$ and $c_3 \cdot n \pm o(n)$ users of degree $3$, what is the maximum $c_4$ such that we can add $c_4 \cdot n \pm o(n)$ users of degree $4$ and  match a $1-o(1)$ fraction of users? And so on and so forth.

\smallskip

In general, we  set $c_d = \frac1{d \cdot (d-1)}$ %
for each degree $d \ge 2$, and $c_0 = c_1 = 0$. %
These fractions result in a total of $n$ users, %
to be matched to the $n$ ads.
To establish that these $c_d$'s give rise to a graph admitting a quasi-complete matching we couple the  irregular cuckoo hashing graph with a random graph in the configuration model having  the same degree distribution on the user side, and the same Poissonian degree distribution  on the ads  side; then, by means of the Karp-Sipser algorithm \cite{ks81} (which we analyze with a variant of the techniques presented in~\cite{bg15} by Balister and Gerke), we establish that this configuration model graph admits a quasi-complete matching; using the coupling, we derive the existence of a quasi-complete matching in the irregular cuckoo hashing graph, as well. %
Finally, we upper bound the fraction of ads matched by any online algorithm with $1- \frac e{e^e} = 0.82062\ldots$, thus proving our new upper bound on the competitive ratio of several Online Bipartite Matching problems. 

\smallskip

As an extra step, aimed at establishing the {\em extremality} of our construction, we show that if one increases any of the $c_d$'s by any constant  $\epsilon > 0$, then a constant fraction of the vertices of degree at most $d$ become unmatchable. (In particular, if the extra $\epsilon$ fraction assigned to degree $d$ is taken from vertices of higher degree, then the quasi-complete matching disappears).

\smallskip

\begin{figure}[t]
\includegraphics[width=\textwidth]{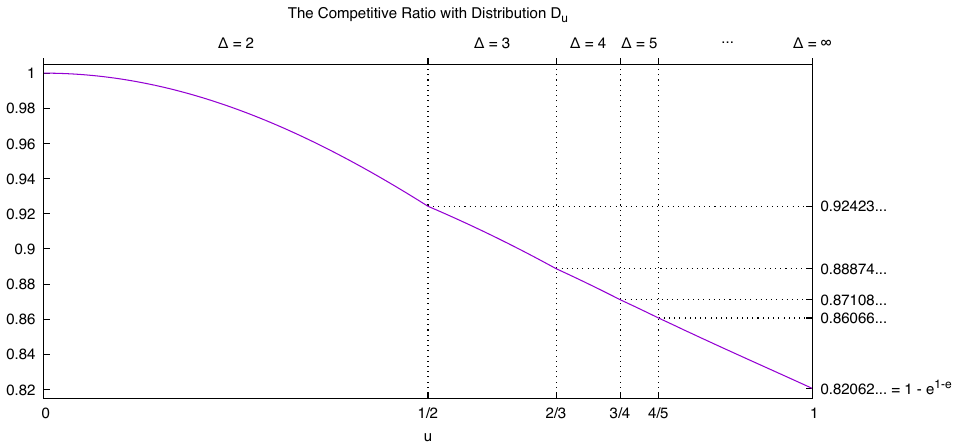}
\caption{The Competitive Ratios induced by the $D_u$ instances, obtained numerically for $u > 0$. The plot shows how $u = 1$ (which corresponds to the case of unbounded user-degree) gives the strongest upper bound. In Theorem~\ref{thm:optimal_u}, we  analytically establish the optimality of $u = 1$.\label{fig:DuCR}}
\end{figure}

{\bf A Class of Constructions.} Our construction has users of unbounded  degree, and contains the same number of users and ads. To get a more thorough understanding of the construction, and of the problem, we  study a class of constructions that generalizes our main one; the instances in this class have $n$ ads, and  $u \cdot n$ users. Just like  our main construction, we  populate the generic instance in the class by  greedily adding users with the smallest possible degree, under the constraint that the instance admits a quasi-complete matching. For a given $u \in (0,1)$, the generalized instance $D_u$ has a maximum user degree of $\left\lceil\frac1{1-u}\right\rceil$. We prove that, for all $u \in (0,1)$, these instances give a weaker upper bound on the competitive ratio than the case  $u = 1$, that is,  the case with $n$ users and $n$ ads (see Figure~\ref{fig:DuCR} and Theorem~\ref{thm:optimal_u}).

\smallskip

{\bf Previous Constructions.} We point out that studying  instances of varying user degrees is necessary to improve over the state of the art. In particular, if an irregular cuckoo hashing instance has $n$ users and $n$ ads, and all the users have degree $1$, any greedy online algorithm will necessarily return a maximum matching (which matches a fraction of $1 - \frac1e = 0.63212\ldots$ users). Thus, the best upper bound provable with this instance is $1$. If, instead, each user has degree $2$, a greedy online algorithm will return a matching with, roughly, $\tanh(1) \cdot n = 0.76159\ldots \cdot n$ edges.         The maximum matching  has size approximately equal to $\left(1 + \frac12 W\left(-2 e^{-2}\right)+\frac14 W\left(-2 e^{-2}\right)^2\right) \cdot n = 0.83809\ldots \cdot n$ (see, e.g., \cite{khk11}; here, $W$ is the Lambert W function), so that the best upper bound provable with this instance is $0.90871\ldots$. Finally,  if each user has degree $3$ or more, one can show that the number of  users matched by any greedy online algorithm is at least $0.82304\ldots \cdot n$. 
It follows that having a uniform user degree in  equibipartite graphs makes it impossible to improve upon---or even match---the  $0.823$ upper bound of Manshadi et al.~\cite{mos11}.

\smallskip

Indeed, Manshadi et al.~\cite{mos11} derived their upper bound by employing a mixture of users with varying degrees. In their construction, the user degree distribution is supported on $\{2,3,n\}$.\footnote{Specifically, Manshadi et al.~\cite{mos11} have $n$ ads and $n$ users, and assign probability $0.40517$ to user-degree $2$, $0.40517$ to user-degree $3$, and $0.18966$ to user-degree $n$. Using the results in~\cite{dgmmpr10}, they prove that the graph admits a maximum matching of size $n - o(n)$. They conclude their argument by numerically computing the size of the online matching, which they find to be no larger than $0.823n$.%
} Manshadi et al.'s analysis leverages on an irregular cuckoo hashing construction of Dietzfelbinger et al.~\cite{dgmmpr10}, whose parameters were obtained numerically; %
as a result, the optimal fractions of vertices of degrees $2$ and $3$ %
in the bipartite matching result of~\cite{mos11} are not given explicitly (and might very well be irrational numbers). %
Our extremal construction, instead, allows us to compute the competitive ratio analytically and exactly, and is made up of simple, rational, probabilities.

\section{Related Work}\label{sec:relatedwork}

Online Bipartite Matching problems have a long history. We mention here the results that are most relevant to our paper; see the book of Mehta~\cite{mehta13}, or the recent survey of Huang, Tang and Wajc~\cite{htw24}, for a more comprehensive introduction to these problems. 

As we mentioned before, if the graph and the ordering are adversarial, then a competitive ratio of $1 - \frac1e$ can be achieved via the RANKING algorithm of Karp, Vazirani and Vazirani \cite{kvv90} (see, Goel and Mehta~\cite{gm08}, or Birnbaum and Mathieu~\cite{bm08}, for proofs; simpler analyses were given in, e.g., Devanur, Jain and Kleinberg~\cite{djk13} and Eden et al.~\cite{effs21}): this algorithm begins by sampling a uniform-at-random ordering of the ads; whenever a user arrives, if at least one ad is available for that user, the user is matched to the first available ad in the ordering. As shown in~\cite{kvv90}, the competitive ratio of the RANKING algorithm, $1- \frac1e$, is optimal in the fully-adversarial %
case.

\smallskip

Determining the optimal competitive ratio if the set of users is permuted uniformly-at-random is, possibly, the most important open problem in this area.
The best known upper bound on the competitive ratio of this problem was proved by Manshadi, Oveis Gharan and Saberi~\cite{mos11,mos12} and has value $0.823$ --- this upper bound holds regardless of whether the bipartite graph is generated adversarially or randomly. The RANKING algorithm is still the best known algorithm  for this  problem  on adversarial graphs; Mahdian and Yan\cite{my11} show that it achieves a competitive ratio of at least $0.696$, and Karande, Mehta and Tripathi~\cite{kmt11} show that its competitive ratio is not larger than $0.727$.

\smallskip

The case of randomized bipartite graphs has also been studied extensively;
in this setting the adversary produces a distribution over subsets of ads, and each user samples independently a set of ads from that distribution. In the ``known distribution'' case, this distribution is assumed to be known to the algorithm; the ``unknown distribution'' variant has also been studied extensively. As  mentioned in the previous paragraph, Mahdian and Yan~\cite{my11} gave an algorithm for randomly-permuted adversarial graphs  with a competitive ratio of $0.696$; this algorithm applies to both the known, and the unknown, distribution cases.
Several improvements on the competitive ratio for the known distribution case  have been proved in the last fifteen years.
In 2011, Manshadi et al.~\cite{mos11},  proposed  an algorithm with a competitive ratio of $0.702$. In 2014, Jaillet and Lu\cite{jl14} had an algorithm achieving a competitive ratio of $0.706$.
Then, in 2021, Huang and Shu~\cite{hs21} gave an algorithm with a competitive ratio of $0.711$; finally,  in 2022, Huang, Shu and Yan\cite{hsy22} achieved a competitive ratio of $0.716$. %

\smallskip

The best impossibility result so far --- which holds in the cases of adversarial graph with UAR ordering, of random graphs with unknown distribution, and known distribution, as well as in the irregular cuckoo hashing case --- is the aforementioned upper bound of $0.823$ published by Manshadi et al. in 2011~\cite{mos11}.
Our result improves the bounds in each of these four cases from $0.823$ to $1-\frac e{e^e} = 0.82062\ldots$. We remark that our numerical improvement, of value roughly $0.002$, is of the same order of all the algorithmic improvements mentioned in the previous paragraph.

\smallskip

Several special cases of online bipartite matching problems have been studied in the last few years. In particular, the ``integral arrival rates'' special case of the ``known distribution'' setting requires that each set in the support of the distribution has an integral expected number of occurrences.\footnote{We observe that our construction, like the  $0.823$ construction of Manshadi et al.~\cite{mos11}, is an equibipartite irregular cuckoo hashing instance with a maximum user degree larger than $1$. Consequently, the total number of ``user types'' in the support of the distribution is at least quadratic in the number of users. Thus, these two constructions do not exhibit integral arrival rates.} For this case, the best known algorithm --- due to Brubach et al. \cite{bssx16} --- achieves a competitive ratio of $0.7299$; the best known upper bound for the integral arrival rates case has value $1-e^{-2} = 0.864\ldots$, and is due to Manshadi et al.~\cite{mos12}.

\smallskip

We conclude this section by noting that the interest in online bipartite matching is further evidenced by the attention given in recent years to several related problems. These include some of its weighted \cite{hsy22,knr22,fhtz22}, and fairness-aware \cite{mx24},  variants; in particular, %
Huang et al. \cite{hsy22} show an upper bound of $0.703$ for the (known distribution) edge-weighted online bipartite matching problem, whereas Ma and Xu \cite{mx24} show an upper bound of $\sqrt{3} - 1 = 0.732\ldots$ for a variant whose goal is to maximize the minimum matching rate across different groups of users. Other related problems include fractional variants and their rounding schemes \cite{tz24, nsw25, ss18, bnw23,bsvw24}, variants with different probing models \cite{bm25,bmr21,bmr22}, those with time-dependent user distributions \cite{bdpsw25,tww22,ppsw21,bdm22}, with users arriving in batches \cite{fns21,fns24,jm22} or with different arrival models \cite{bst19,gkmsw19,efgt22}, variants that leverage machine-learned advice \cite{cglb24,agkk23,jm22}, as well as generalized selection tasks \cite{hjpsz24,ghhnyz21}.

\section{Preliminaries}\label{sec:preliminaries}

We present here the objects, the tools and the notation that we will use in this paper.

\smallskip

{\bf  Notation.} Given a predicate $P$, the Iverson bracket $[P]$ has value $1$ if $P$ is true, and value $0$ otherwise. With a slight abuse of notation, for a positive integer $n$, we use $[n]$ to denote the set $[n] = \{1,\dots, n\}$.
The support  of a multiset $S$ is the set of elements whose multiplicity in $S$ is at least $1$.
 For an integer $k\geq 1$, we let $H_k=\sum_{i=1}^k \frac1i$ be the $k$-th harmonic number. Finally, we let $\Phi(z,s,\alpha)=\sum_{k=0}^\infty \frac{z^k}{(k+\alpha)^s}$ be the Lerch transcendent, a series that generalizes %
 the Riemann $\zeta$ function, %
and which converges if $\alpha > 0$, $|z| < 1$ and $s \in \mathbf{R}$. 

\smallskip

{\bf Bipartite (Multi)Graphs and Matchings.} An undirected bipartite (multi)graph $G(V,\hat{V},E)$ is defined by two disjoint sets of vertices $V$ and $\hat{V}$, and by a (multi)set of edges $E$ having a support set which is a subset of $\{\{v,\hat{v}\} \mid v \in V \wedge \hat{v} \in \hat{V}\}$. If $E$ coincides with its support set, then $G(V,\hat{V},E)$ has no parallel edges, and we say that $G(V,\hat{V},E)$ is a simple graph. We will typically use $V=\{v_1, v_2, \dots\}$ to denote the set of users, and $\hat{V}=\{\hat{v}_1, \hat{v}_2, \dots\}$ to denote the set of ads. We will also use $n_i$ (resp., $\hat{n}_i$) to denote the number of vertices of degree $i\geq 0$ in $V$ (resp., $\hat{V}$). 

\smallskip

A matching $M$ in the (multi)graph $G(V,\hat{V},E)$ is a subset $M \subseteq E$ containing pairwise disjoint edges.
We say that $G(V,\hat{V},E)$ admits a {\em complete matching} if there exists a matching $M \subseteq E$ of $G(V,\hat{V},E)$ such that $|M| = \min(|V|,|\hat{V}|)$, and that it admits a {\em perfect matching} if there exists a matching $M \subseteq E$ such that $|M| = |V| = |\hat{V}|$.

\smallskip

Given a bipartite multigraph $G(V,\hat{V},E)$, if $E'$ is the support set of $E$, we have that $G(V,\hat{V},E')$ is a bipartite simple graph. %
Observe that each matching of $G(V,\hat{V},E')$ is also a matching of $G(V,\hat{V},E)$; moreover, given an arbitrary matching $M$ of $G(V,\hat{V},E)$, we can transform $M$ into a matching of $G(V,\hat{V},E')$ by substituting, for each $e \in M$, the edge $e$ with its unique parallel edge $e' \in E'$.

\smallskip

We  say that an infinite sequence of bipartite (multi)graphs $G(V_1,\hat{V}_1,E_1), G(V_2,\hat{V}_2,E_2), \ldots$ of increasing side sizes ($\min(|V_1|,|\hat{V}_1|) < \min(|V_2|,|\hat{V}_2|) < \cdots$), and with maximum matchings $M_1 \subseteq E_1, M_2 \subseteq E_2,\ldots$,   admits a {\em quasi-complete matching} if $|M_i| \ge (1-o(1)) \cdot \min(|V_i|, |\hat{V}_i|)$ --- that is, if for each $\epsilon > 0$ there exists $i^{\star}$ such that, for each $i \ge i^{\star}$, $|M_i| \ge (1-\epsilon) \cdot \min(|V_i|, |\hat{V}_i|)$.

\smallskip

{\bf The Karp-Sipser Algorithm.}
The Karp-Sipser algorithm~\cite{ks81} is a greedy algorithm for computing a matching of a graph. A run of this algorithm is  divided in two phases; for our purposes, only the first phase will be necessary. During the first phase, the algorithm looks for any vertex of degree $1$; as long as such a vertex exists, the algorithm adds its incident edge $e$ to a set $S$,  removes the two endpoints of $e$ from the graph, and iterates. If no such vertex exists, the first phase ends. It is easy to prove that, at any point during the execution of the first phase, the set $S$ can be extended to a maximum matching of the original graph. In fact, it is also easy to prove that if the original graph $G(V,\hat{V},E)$ is bipartite and, at any point during the execution of the first phase, $V_0 \subseteq V$ (resp. $\hat{V}_0\subseteq \hat{V}$) is  the subset of $V$ (resp. $\hat{V}$) containing  vertices that  have current degree $0$ and which are still unmatched, then the maximum matching cannot be larger than $\min\left(|V| - |V_0|, |\hat{V}| - |\hat{V}_0|\right)$.%

\smallskip

The Karp-Sipser algorithm was introduced to bound the size of the maximum matching in sparse $G(n,p)$ random graphs; since its introduction, the algorithm has been used
to bound the size of the maximum matching of many other random graph models.

\smallskip

{\bf The Irregular Cuckoo Hashing Model.} For simplicity of exposition, the random graph models that we will consider in this paper %
might produce non-simple graphs, that is, multigraphs with parallel edges; these can be easily  transformed into simple graphs through the removal of parallel edges. One of the two  random graph models that we will consider  is the irregular cuckoo hashing model, which is defined as follows:
\begin{definition}[Irregular Cuckoo Hashing Model]
Let $n, m$ be two positive integers, and $P$ be a probability distribution over the non-negative integers. The Irregular Cuckoo Hashing random multigraph $\irregularcuckoohashing(n,m,P)$ is defined as follows. 
This bipartite graph has user set $V$ and ads set $\hat{V}$, with $n = |V|$ and $m = |\hat{V}|$. Each user $v_i \in V$ samples independently an integer $d_i$ from $P$ and, for each $j = 1,2,\ldots,d_i$, samples independently an ad $\hat{v}_{i,j}$ uniformly at random from $\hat{V}$. The multiset of neighbors of $v_i$ is  $\{\hat{v}_{i,1},\ldots,\hat{v}_{i,d_i}\}$. (Note that $\irregularcuckoohashing(n,m,P)$ might not be simple).
\end{definition}
\smallskip

{\bf The Configuration Model.} In order to analyze the irregular cuckoo hashing graph of our construction, we will  couple it with a random graph in the configuration model.
\begin{definition}[Configuration Model]
Consider two sequences of non-negative integers $\mathbf{d}=(d_1, \dots, d_n)$ and $\mathbf{\hat{d}}=(\hat{d}_1, \dots, \hat{d}_m)$ such that $\sum_{i=1}^n d_i = \sum_{i=1}^m \hat{d}_i$. To sample a bipartite graph $G=(V, \hat{V}, E)$ from the configuration model $\CM(\mathbf{d}, \mathbf{\hat{d}})$, for each $i\in[n]$ (resp. $j\in [m]$), assign to vertex $v_i$ (resp. $\hat{v}_j$) $d_i$ (resp. $\hat{d}_j$) configuration points, and then sample a uniform at random perfect bipartite matching between the configuration points. (Note that $\CM(\mathbf{d}, \mathbf{\hat{d}})$ might not be simple).
\end{definition}
To compute the size of the maximum matching in our configuration model distribution, we will study the behavior of the first phase of the Karp-Sipser algorithm using the techniques presented, e.g., by Balister and Gerke in~\cite{bg15} (see also \cite{afp98,bf11}).
In particular, we will show that our distribution admits a quasi-complete matching through the following  modification\footnote{The Theorems of Balister and Gerke require the degree distributions to have finite support. Observe that any degree distribution with a finite mean can be truncated to a finite prefix  altering the distribution's expectation  by no more than any  $\epsilon > 0$; in the setting of the configuration model, this truncation results in a change of no more than an $\epsilon$ fraction in the maximum matching size. Balister and Gerke use this approach to apply their Theorems to generic distributions of finite average degree. In our case, however, the ads degree distribution makes it possible to apply several algebraic shortcuts that become infeasible once that distribution is transformed into one having finite support. Therefore, to simplify our analysis, we reproved certain Lemmas and Theorems of Balister and Gerke~\cite{bg15}, so to accommodate distributions with unbounded support.} of a result in \cite{bg15}:
\begin{theorem}\label{thm:bg}%
Let $\rho\in(0,1]$ be any constant. Let $f(x)=\sum_{i=0}^\infty z_i\cdot x^i$, and $\hat{f}(x) = \sum_{i=0}^\infty \hat{z}_i\cdot x^i$, where $(z_i)_{i=0}^\infty, (\hat{z}_i)_{i=0}^\infty$ are non-negative sequences such that $f'(1)=\hat{f}'(1)=\mu < \infty$ and $f'(2), \hat{f}'(2) < \infty$.
Consider any configuration model $\CM(\mathbf{d}, \mathbf{\hat{d}})$, with $|\mathbf{d}|=f(1)n\pm o(n)$,  $|\mathbf{\hat{d}}|=\hat{f}(1)n \pm o(n)$, and such that %
(i) $n_i=(z_i\pm o(n^{-\rho}))n$ for $i\geq 0$,
(ii) $\hat{n}_i=(\hat{z}_i \pm o(n^{-\rho}))n$ for $i\geq 0$,
(iii) $|E|=(\mu \pm o(n^{-\rho}))n$,
and (iv) for $\maxdeg=\maxdeg(n)=\Theta(\log n)$, $\sum_{i=\maxdeg}^\infty i^2 \cdot z_i = o(n^{-\rho})$ and $\sum_{i=\maxdeg}^\infty i^2 \cdot \hat{z}_i = o(n^{-\rho})$,
where $n_i$ (resp. $\hat{n}_i$) is the number of vertices of degree $i$ in $\mathbf{d}$ (resp. $\mathbf{\hat{d}}$). 

Let $w_2, \hat{w}_1 \in [0,1]$ be the smallest solutions to the simultaneous equations $w_2=1 - \frac{f'(1-\hat{w}_1)}{\mu}$, $\hat{w}_1 = \frac{\hat{f}'(w_2)}{\mu}$. Then, with probability $1-o(1)$, the first phase of Karp-Sipser, on the graph $G(V, \hat{V},E)\sim\CM(\mathbf{d}, \mathbf{\hat{d}})$, returns a matching of size at least
\[
(f(1) - f(1-\hat{w}_1) - f'(1-\hat{w}_1)\hat{w}_1 - o(1))n.
\]
Moreover, with probability $1-o(1)$, after the first phase of Karp-Sipser, at least $(\hat{f}(w_2) - f(1) + f(1-\hat{w}_1)+f'(1-\hat{w}_1)\hat{w}_1 - o(1))n$ vertices in $\hat{V}$ will be isolated (i.e., of degree 0 and unmatched). Equivalently, with probability $1-o(1)$, the maximum matching cannot be larger than,
\[
(\hat{f}(1) - \hat{f}(w_2) + f(1) - f(1-\hat{w}_1) - f'(1-\hat{w}_1)\hat{w}_1 + o(1))n.
\]
\end{theorem}
\begin{toappendix}
    \section{Adaptation of the Analysis of Balister and Gerke to Unbounded Degree Distributions}\label{app:bg}
In this section we adapt the analysis of \cite{bg15} for the first phase of the Karp-Sipser algorithm, so to allow unbounded maximum degree. Note that the original analysis can be used by first limiting the maximum degree to a constant and by truncating the generating functions $f$ and $\hat{f}$ defined below;
in our case, however, the original distribution makes it possible to apply several algebraic shortcuts that become infeasible once the distribution is transformed into one having finite support. Therefore, we found it simpler to modify Balister and Gerke's analysis so to  accommodate  distributions with unbounded support.

\smallskip

The analysis of \cite{bg15} considers a slight variation of the first phase of Karp-Sipser. The algorithm proceeds in rounds; in particular, all degree-1 vertices of round $t$ are analyzed before moving to round $t+1$. The degree-1 vertices of round $t+1$ are those that were generated in round $t$. This slight modification still has the property that, at any point, the current matching can be extended to a maximum one. 

Define the generating functions $f(x)=\sum_{i=0}^\infty z_i \cdot x^i$ and $\hat{f}(x)=\sum_{i=0}^\infty \hat{z}_i\cdot x^i$, where $(z_i)_{i=0}^\infty$ and $(\hat{z}_i)_{i=0}^\infty$ are two non-negative sequences, and we adopt the standard convention $0^0=1$. We assume that $f'(1)=\hat{f}'(1)=\mu < \infty$, and $f'(2) = \sum_{i=1}^\infty z_i \cdot i \cdot 2^{i-1} < \infty$, $\hat{f}'(2) = \sum_{i=1}^\infty \hat{z}_i \cdot i \cdot 2^{i-1} < \infty$. Note that this also implies $f''(1)<\infty$, $\hat{f}''(1)<\infty$, $\sum_{i=1}^\infty z_i \cdot i \cdot 2^i < \infty$, and $\sum_{i=1}^\infty \hat{z}_i \cdot i \cdot 2^i < \infty$. Consider any configuration model $\CM(\mathbf{d}, \mathbf{\hat{d}})$, with $|\mathbf{d}|=f(1)n\pm o(n)$,  $|\mathbf{\hat{d}}|=\hat{f}(1)n \pm o(n)$, and such that:
\begin{enumerate}
\item $n_i=(z_i\pm o(n^{-\rho}))n$ for $i\geq 0$
\item $\hat{n}_i=(\hat{z}_i \pm o(n^{-\rho}))n$ for $i\geq 0$
\item $|E|=(\mu \pm o(n^{-\rho}))n$
\item for $\maxdeg=\maxdeg(n)=\Theta(\log n)$, $\sum_{i=\maxdeg}^\infty i^2 \cdot z_i = o(n^{-\rho})$ and $\sum_{i=\maxdeg}^\infty i^2 \cdot \hat{z}_i = o(n^{-\rho})$
\end{enumerate}
where $\rho\in(0,1]$ is a constant, and $n_i$ (resp. $\hat{n}_i$) is the number of vertices of degree $i$ in $\mathbf{d}$ (resp. $\mathbf{\hat{d}}$).

As a first step, similarly to \cite{bg15}, we limit the maximum degree to $\maxdeg$. Note that, however, we will not modify the definitions of $f$ and $\hat{f}$. Specifically, we remove all configuration points of the vertices having degree larger than $\maxdeg$. Therefore, on one side of the bipartite graph, we remove a number of configuration points equal to 
\begin{align*}
\sum_{i=\maxdeg+1}^\infty i n_i &= |E| - \sum_{i=0}^\maxdeg i n_i = n\sum_{i=0}^\infty iz_i - n\sum_{i=0}^\maxdeg i z_i\pm o(n^{1-\rho} \maxdeg^2) \leq o(n^{1-\rho} \maxdeg^2) \leq o(n^{1-\rho'}),
\end{align*}
for some new constant $0<\rho'<\rho$. After removing $o(n^{1-\rho'})$ configuration points from the other side of the bipartite graph, and further removing $o(n^{1-\rho'})$ configuration points to ensure that the sum of the degrees remains the same on the two sides; we end up with a configuration model having a maximum degree of $\maxdeg=O(\log n)$, but since we only removed $o(n^{1-\rho'})$ edges, properties 1-4 above still hold (albeit with a different constant $\rho$) and the size of the maximum matching is the same up to an $o(n)$ error. For simplicity, we will keep on using $\rho$ in place of $\rho'$.

Let us sample a graph $G(V,\hat{V},E)$ from such configuration model. A vertex $v\in V\cup \hat{V}$ is called $t$-good if its $t$-neighborhood (i.e., $v$ and all the vertices at distance at most $t$ from $v$) is a tree, otherwise it is called $t$-bad. (Note that a vertex with parallel edges is also considered $t$-bad.) The following result only relies on $\maxdeg^{t} = c^{o(\log n)}$ for some constant $c$.\footnote{In \cite{bg15} they have $t=o(\log n)$, we need to decrease it to $o\left(\frac{\log n}{\log\log n}\right)$ since our maximum degree is super-constant.}

\begin{lemma}[\hspace{1sp}\cite{bg15}, Lemma 3.4]
If $t=o\left(\frac{\log n}{\log\log n}\right)$, with probability $1-o(1)$, the number of $t$-bad vertices is $o(n)$.
\end{lemma}

Fix a $t$-good vertex $v$. Consider the tree rooted at $v$. Starting from $v$, we number the levels of the tree from 1 to $t+1$. Consider the vertices different from the root, we classify them as follows. All vertices at level $t+1$ are $v$-normal. Consider a vertex $u$ at level $\ell$. If all its children at level $\ell+1$ are $v$-popular, then $u$ is $v$-lonely (this considers also the case where there are no children at level $\ell+1$). If at least one of its children is $v$-lonely, then $u$ is $v$-popular. Otherwise, $u$ is $v$-normal. The root $v$ is lonely wrt $t$ if all its children are $v$-popular (or it has no children), it is popular wrt $t$ if at least two of its children are $v$-lonely, otherwise it is normal wrt $t$. The following result only uses the classification of the vertices.

\begin{lemma}[\hspace{1sp}\cite{bg15}, Lemma 3.5]\label{lem:bg-3-5}
Suppose $v$ is $t$-good, for $t\geq 1$. If $v$ is popular wrt $t$, it will be matched within the first $\lfloor t/2 \rfloor$ rounds of the first phase of the Karp-Sipser Algorithm. If $v$ is lonely wrt $t$, in the first $\lfloor t/2\rfloor$ rounds, it will either become isolated (i.e., of degree $0$ and unmatched) or it will be matched to a vertex $u$ which is popular wrt $t-1$.
\end{lemma}

Define inductively, $\hat{w}_1^{(1)} = w_2^{(1)} = 0$, and $\hat{w}_1^{(t)} = \frac{\hat{f}'(w_2^{(t-1)})}{\mu}$, $w_2^{(t)} = 1 - \frac{f'(1 - \hat{w}_1^{(t-1)})}{\mu}$ for $t>1$. Let $(\hat{w}_1, w_2)$ be the smallest non-negative solution of the simultaneous equations $\hat{w}_1 = \frac{\hat{f}'(w_2)}{\mu}$, $w_2 = 1 - \frac{f'(1 - \hat{w}_1)}{\mu}$. As $t$ goes to infinity, $(\hat{w}_1^{(t)}, w_2^{(t)})$ converges to $(\hat{w}_1, w_2)$, since it is a bounded increasing sequence. Indeed, the next Lemma shows that $\hat{w}_1^{(i)}$ and ${w}_2^{(i)}$ are probabilities. We need to adapt the original proof of \cite{bg15} for this Lemma because we use the original, non-truncated, generating functions $f$ and $\hat{f}$. Specifically, we will use that $\sum_{i=0}^M z_i \cdot x^i = f(x) \pm o(n^{-\rho})$ for $x\in[0,1]$, and similarly for $f'$, $f''$, $\hat{f}$, $\hat{f}'$, and $\hat{f}''$.

\begin{lemma}[Adaptation of \cite{bg15}, Theorem 3.6]\label{lem:bg-adaptation-3-6}
Let $t=o\left(\frac{\log n}{\log\log n}\right)$, and let $v$ be a $t$-good vertex. For all $i\in\{1,\dots, t\}$, the following holds:
\begin{enumerate}
\item The probability that a vertex $u\in \hat{V}$ at distance $t+1-i$ from $v$ is $v$-lonely wrt $t$, conditioned on its existence and the tree rooted at $v$ except for the branch rooted at $u$, is $\hat{w}_1^{(i)} \pm O(\maxdeg^{t+i} \cdot n^{-\rho})$
\item The probability that a vertex $u\in V$ at distance $t+1-i$ from $v$ is $v$-popular wrt $t$, conditioned on its existence and the tree rooted at $v$ except for the branch rooted at $u$, is $w_2^{(i)} \pm O(\maxdeg^{t+i} \cdot n^{-\rho})$
\end{enumerate}
The results are true even conditioning on at most $O(\maxdeg^t)$ edges outside of the $t$-neighborhood of $v$.
\end{lemma}
\begin{proof}
Both results are true for $i=1$, as all vertices at level $t+1$ are $v$-normal. Let $i>1$. We start by proving the first point. Fix a vertex $u\in \hat{V}$ at distance $t+1-i$, the probability that it is $v$-lonely is:
\[
\sum_{k=0}^{\maxdeg-1} \Pr[u\text{ has children }u_1,\dots,u_k] \prod_{j=1}^k \Pr[u_j\text{ is $v$-popular} \mid u_1,\dots,u_{j-1}\text{ are $v$-popular}].
\]
Considering the conditioning of at most $O(\maxdeg^t)$ edges, the probability that $u$ has degree $k+1$ is:
\begin{align*}
\frac{(k+1)(\hat{z}_{k+1}\pm o(n^{-\rho}))n - O(\maxdeg^t)}{(\mu\pm o(n^{-\rho}))n - O(\maxdeg^t)} &= \frac{(k+1)\hat{z}_{k+1}}{\mu} \pm O\left(\frac{\maxdeg^t}{n}\right) \pm o\left(\frac{k}{n^{\rho}}\right) = \frac{(k+1)\hat{z}_{k+1}}{\mu} \pm O\left(\frac{\maxdeg^t}{n^{\rho}}\right).
\end{align*}

Now, the probability of $u_j$ being $v$-popular is by inductive hypothesis $w_2^{(i-1)} \pm O(\maxdeg^{t+i-1} n^{-\rho})$. Note also that $\left(w_2^{(i-1)} \pm O(\maxdeg^{t+i-1} n^{-\rho})\right)^k = (w_2^{(i-1)})^k \pm k\cdot O(\maxdeg^{t+i-1} n^{-\rho})$, since we can discard lower order terms. Putting things together, the previous probability becomes:

\begin{align*}
&\sum_{k=0}^{\maxdeg-1} \frac{(k+1)\hat{z}_{k+1}}{\mu} \cdot \left(w_2^{(i-1)} \pm O\left(\frac{\maxdeg^{t+i-1}}{ n^{\rho}}\right)\right)^k \pm O\left(\frac{
\maxdeg^{t+1}}{n^\rho}\right)\\
& = \sum_{k=0}^{\maxdeg-1} \frac{(k+1)\hat{z}_{k+1}}{\mu} \cdot \left(w_2^{(i-1)}\right)^k \pm O\left(\frac{\maxdeg^{t+i-1}}{ n^{\rho}}\right) \sum_{k=0}^{\maxdeg-1} \frac{k(k+1)\hat{z}_{k+1}}{\mu} \pm O\left(\frac{
\maxdeg^{t+1}}{n^\rho}\right)\\
& = \frac1\mu \left( \sum_{k=1}^\infty k\hat{z}_{k} \left(w_2^{(i-1)}\right)^{k-1} -\sum_{k=\maxdeg+1}^\infty k\hat{z}_{k} \left(w_2^{(i-1)}\right)^{k-1}\right) \pm O\left(\frac{
\maxdeg^{t+1}}{n^\rho}\right) \\
&\quad \pm O\left(\frac{\maxdeg^{t+i-1}}{ n^{\rho}}\right) \frac1\mu \left( \sum_{k=2}^{\infty} (k-1)k\hat{z}_{k} - \sum_{k=\maxdeg+1}^{\infty} (k-1)k\hat{z}_{k} \right)\\
&= \frac1\mu \left( \hat{f}'(w_2^{(i-1)}) -o(n^{-\rho})\right) \pm O\left(\frac{\maxdeg^{t+i-1}}{ n^{\rho}}\right) \frac1\mu \left( \hat{f}''(1) - o(n^{-\rho}) \right) \pm O\left(\frac{
\maxdeg^{t+1}}{n^\rho}\right)\\
& = \frac{\hat{f}'(w_2^{(i-1)})}{\mu} \pm O\left(\frac{
\maxdeg^{t+i}}{n^\rho}\right)\\
& = \hat{w}_1^{(i)} \pm O\left(\frac{
\maxdeg^{t+i}}{n^\rho}\right),
\end{align*}
where we used that $\hat{f}''(1) \leq O(\maxdeg)$. This proves the first point. The second point is proved similarly. Fix a vertex $u\in V$ at distance $t+1-i$. By inductive hypothesis, the probability that a vertex $u_j$ at distance $t+1-(i-1)$ is $v$-lonely is $\hat{w}_1^{(i-1)} \pm O(\maxdeg^{t+i-1}n^{-\rho})$. The probability that $u$ is $v$-popular is:
\begin{align*}
&1 - \sum_{k=0}^{\maxdeg-1} \Pr[u\text{ has children }u_1,\dots,u_k] \prod_{j=1}^k (1 - \Pr[u_j\text{ is $v$-lonely} \mid u_1,\dots,u_{j-1}\text{ are not $v$-lonely}])\\
&= 1 - \frac1\mu \sum_{k=0}^{\maxdeg-1} (k+1)z_{k+1} \cdot \left(1 - \hat{w}_1^{(i-1)} \pm O\left(\frac{\maxdeg^{t+i-1}}{n^{\rho}}\right)\right)^k \pm O\left(\frac{\maxdeg^{t+1}}{n^{\rho}}\right)\\
& = 1 - \frac1\mu \sum_{k=0}^{\maxdeg-1} (k+1)z_{k+1} \cdot (1-\hat{w}_1^{(i-1)})^k \pm O\left(\frac{\maxdeg^{t+i-1}}{n^{\rho}}\right)\frac1\mu \sum_{k=0}^{\maxdeg-1} k(k+1)z_{k+1}\pm O\left(\frac{\maxdeg^{t+1}}{n^{\rho}}\right)\\
& = 1 - \frac{f'(1-\hat{w}_1^{(i-1)})}{\mu} + o(n^{-\rho}) \pm O\left(\frac{\maxdeg^{t+i-1}}{n^{\rho}}\right)\frac1\mu\left(f''(1) - o(n^{-\rho}) \right) \pm O\left(\frac{\maxdeg^{t+1}}{n^{\rho}}\right)\\
&=w_2^{(i)} \pm O\left(\frac{\maxdeg^{t+i}}{n^{\rho}}\right). \qedhere
\end{align*}
\end{proof}

The following result only depends on the definition of popular/lonely and on the properties of the first phase of Karp-Sipser.
\begin{lemma}[\hspace{1sp}\cite{bg15}, proof of Theorem 3.7]\label{lem:bg-concentration-3-7}
Let $t=o\left(\frac{\log n}{\log \log n}\right)$. Let $X$ be either the number of $t$-good vertices $v\in V$ that are popular, or the number of $t$-good vertices $u\in \hat{V}$ that are lonely. In both cases, it holds $\Var[X] \leq O(n\cdot \maxdeg^{2t+1})$
\end{lemma}

The following Lemma again requires a slight adaptation to account for the non-truncated generating functions. 

\begin{lemma}[Adaptation of \cite{bg15}, Theorem 3.7]\label{lem:bg-3-7}
Let $t=o\left(\frac{\log n}{\log\log n}\right)$. As $t$ goes to infinity, with probability $1-o(1)$, 
\begin{enumerate}
\item the number of $t$-good vertices $v\in V$ that are popular w.r.t. $t$ is $(f(1) - f(1-\hat{w}_1) - f'(1-\hat{w}_1)\cdot \hat{w}_1 \pm o(1))\cdot n$ 
\item the number of $t$-good vertices $v\in \hat{V}$ that are lonely w.r.t. $t$ is $(\hat{f}(w_2) \pm o(1))\cdot n$
\end{enumerate}
\end{lemma}
\begin{proof}
Recall that only $o(n)$ vertices are not $t$-good. A $t$-good vertex $v\in V$ is popular if at least two of its children are $v$-lonely. By \Cref{lem:bg-adaptation-3-6}, each child is $v$-lonely with probability $h(t)=\hat{w}_1^{(t)} \pm O\left(\frac{M^{2t}}{n^{\rho}}\right)$.\footnote{The probabilities of being $v$-lonely of two children differ by at most $O\left(\frac{M^{2t}}{n^{\rho}}\right)$, for simplicity of exposition we use the same $h(t)$ for all of them.} Therefore, the probability that $v$ is popular is $1 - \left(1 - h(t)\right)^{\deg(v)} - \deg(v) \cdot h(t) \cdot \left(1 - h(t)\right)^{\deg(v)-1}$. Note that this probability is correct for $\deg(v)\geq 0$. Thus, the expected number of popular vertices in $V$ is
\begin{align*}
&\sum_{i=0}^{\maxdeg} n (z_i \pm o(n^{-\rho})) \left(1 - \left(1 - h(t)\right)^{i} - i \cdot h(t) \cdot \left(1 - h(t)\right)^{i-1} \right) \\
&= n\sum_{i=0}^{\maxdeg} z_i\left(1 - \left(1 - h(t)\right)^{i} - i \cdot h(t) \cdot \left(1 - h(t)\right)^{i-1} \right) \pm o(\maxdeg\cdot n^{1-\rho})\\
& = n \left(\sum_{i=0}^\maxdeg z_i - \sum_{i=0}^\maxdeg z_i (1-\hat{w}_i^{(t)})^i - \hat{w}_1^{(t)}\sum_{i=0}^\maxdeg i z_i(1-\hat{w}_1^{(t)})^{i-1}  \pm O\left(\frac{\maxdeg^{2t+3}}{n^\rho}\right)  \right) \pm o(\maxdeg\cdot n^{1-\rho})\\
& = n \left(f(1) - o(n^{-\rho}) - f(1-\hat{w}_1^{(t)}) + o(n^{-\rho})  - \hat{w}_1^{(t)}(f'(1-\hat{w}_1^{(t)}) - o(n^{-\rho}))  \pm O\left(\frac{\maxdeg^{2t+3}}{n^\rho}\right)  \right) \pm o(\maxdeg\cdot n^{1-\rho})\\
& = n \left(f(1) - f(1-\hat{w}_1^{(t)}) - \hat{w}_1^{(t)}f'(1-\hat{w}_1^{(t)}) \pm o(1) \right)\\
&=n \left(f(1) - f(1-\hat{w}_1) - \hat{w}_1f'(1-\hat{w}_1) \pm o(1) \right),
\end{align*}
where the last step uses that $\hat{w}_1^{(t)}=\hat{w}_1\pm o(1)$ and, for $x\in[0,1]$, $f(x\pm o(1))=f(x)\pm o(1)$, $f'(x\pm o(1))=f'(x)\pm o(1)$. Indeed, $f(x \pm o(1)) = f(x) + \sum_{i=1}^\infty z_i \sum_{j=1}^i \binom{i}{j} (\pm o(1))^j x^{i-j}$. For $x\in[0,1]$, we have,
\begin{align*}
\left|\sum_{i=1}^\infty z_i \sum_{j=1}^i \binom{i}{j} (\pm o(1))^j x^{i-j} \right| \leq o(1) \sum_{i=1}^\infty z_i \sum_{j=1}^i \binom{i}{j} \leq o(1) \sum_{i=1}^\infty z_i \cdot 2^i = o(1).
\end{align*}
The result for $f'$ (as well as $\hat{f}$ and $\hat{f}'$) is proved similarly. 

We can similarly find the number of lonely vertices in $\hat{V}$. A $t$-good vertex $v\in\hat{V}$ is lonely is all its children are $v$-popular. By \Cref{lem:bg-adaptation-3-6}, each child is $v$-popular with probability $y(t)=w_2^{(t)} \pm O\left(\frac{\maxdeg^{2t}}{n^{\rho}}\right)$, therefore, $v$ is lonely with probability $y(t)^{\deg(v)}$. Thus, the expected number of lonely vertices in $\hat{V}$ is,
\begin{align*}
\sum_{i=0}^{\maxdeg} n(\hat{z}_i \pm o(n^{-\rho})) \cdot y(t)^i = n\sum_{i=0}^\maxdeg \hat{z}_i (w_2^{(t)})^i \pm o(n) = n(\hat{f}(w_2^{(t)}) \pm o(1)) = n(\hat{f}(w_2) \pm o(1)).
\end{align*}
The concentration around the mean directly follows from \Cref{lem:bg-concentration-3-7} and Chebyshev inequality.
\end{proof}
\begin{proof}[Proof of \Cref{thm:bg}]
Run the first phase of Karp-Sipser for $t_0$ rounds, with $t_0=o\left(\frac{\log n}{\log\log n}\right)$ and $t_0\rightarrow \infty$. From \Cref{lem:bg-3-5} the number of popular vertices in $V$ is a lower bound to the size of the matching found by the first phase of Karp-Sipser. Therefore, \Cref{lem:bg-3-7} concludes the proof of the lower bound. As per the upper bound, by \Cref{lem:bg-3-5}, the difference between the lonely vertices in $\hat{V}$ and the popular vertices in $V$ is a lower bound to the number of vertices that become isolated and unmatched in $\hat{V}$ after the first phase. Substituting the expressions given by \Cref{lem:bg-3-7} concludes the proof.
\end{proof}
\end{toappendix}
In Appendix~\ref{app:bg}, we prove Theorem~\ref{thm:bg} by following the approach in~\cite{bg15}, with  minor modifications  to handle distributions of unbounded support.

\smallskip

{\bf Random Variables and Concentration Inequalities.} We denote with $\Bin(n,p)$ the binomial distribution with success probability $p$ and $n$ trials; with $\Poisson(\lambda)$ the Poisson distribution of parameter $\lambda>0$; and with $\geom(p)$ the geometric distribution with success probability $p$, in particular $\Pr_{X\sim \geom(p)}[X=k]=p(1-p)^{k-1}$ for $k\ge 1$, and $\E_{X\sim\geom(p)}[X]=\frac{1}{p}$. For two random variables $X,Y$ taking values in the natural numbers $\mathbf{N}$, we let $\dTV(X,Y)=\sup_{A\subseteq \mathbf{N}}\left|\Pr[X\in A] - \Pr[Y \in A]\right|$ be their total variation distance. 

\smallskip

In our arguments, we will make use of the following well-known concentration bounds (whose proofs can be found, e.g., in~\cite{dp09}): 
\begin{fact}[Chernoff-Hoeffding's inequality]\label{fact:chernoff-hoeffding}
Let $X_1, \dots, X_n$ be independent random variables such that $X_i\in [a,b]$ almost surely. Let $X=\sum_{i=1}^n X_i$, then, for all $t>0$, $$\Pr[|X - \E[X]| \geq t] \leq 2 \exp\left(- \frac{2\cdot t^2}{n\cdot(b-a)^2}\right).$$
\end{fact}
Additionaly, we will employ McDiarmid's inequality, a generalization of the Chernoff-Hoeffding bound:
\begin{fact}[McDiarmid's inequality]
\label{fact:mcdiarmid}
Let $X_1, X_2, \dots, X_n$ be independent random variables with $X_i\in \mathcal{X}$. Let $f:\mathcal{X}^n \rightarrow \mathbb{R}$ be such that for each $i\in[n]$ and for each $x_1, x_2, \dots, x_n, x'_i \in \mathcal{X}$, it holds $|f(x_1, \dots, x_{i-1}, x_i, x_{i+1}, \dots, x_n) - f(x_1, \dots, x_{i-1}, x'_i, x_{i+1}, \dots, x_n)| \leq d$. Then, for each $t>0$, 
\[
\Pr[|f(X_1, \dots, X_n) - \E[f(X_1, \dots, X_n)]| \ge t] \leq 2 \exp\left(- \frac{2 \cdot t^2}{n\cdot d^2}\right).
\]
\end{fact}
We will also leverage on the following concentration inequality for the sum of geometric random variables:
\begin{fact}[\hspace{1sp}\cite{j18}, Theorem 2.1 and Theorem 3.1]\label{fact:geom-concentration} Let $X_1, \dots, X_n$ be independent random variables with $X_i\sim\geom(p_i)$. Let $X=\sum_{i=1}^n X_i$, and let $p^\star=\min_i p_i > 0$. For each $\lambda\ge 1$, $\Pr[X\ge \lambda \E[X]] \leq e^{-p^\star \E[X] (\lambda-1-\ln\lambda)}$, and for each $\lambda\in(0,1]$, $\Pr[X\le \lambda \E[X]] \leq e^{-p^\star \E[X] (\lambda-1-\ln\lambda)}$.
\end{fact}
Finally, we will use the following  approximation, by means of Poisson variables, of binomial random variables having small expectation.%
\begin{fact}[\hspace{1sp}\cite{s94,ap06}]\label{fact:lecam}
Let $X\sim\Bin(n,p)$, $Y\sim \Poisson(\lambda)$ for $\lambda > 0$, then $\dTV(X,Y) \leq np^2 + |\lambda - np|$.
\end{fact}

\section{Impossibility Result}\label{sec:impossibility}
In this Section we present, and analyze, an irregular cuckoo hashing instance %
of the online bipartite matching problem; by means of this instance, we will show that no online algorithm can achieve a competitive ratio better than $1 - \frac e{e^e} = 0.8206259\ldots$.  In Section~\ref{sec:tightness} we will show how this construction can be obtained by greedily maximizing the probabilities of the degrees $0,1,2,3,\ldots$, while guaranteeing that the resulting graph has a quasi-complete matching.
\begin{definition}[Main Instance]\label{def:instance}
Our main instance is given by the irregular cuckoo hashing distribution $\irregularcuckoohashing(n, n, D)$, where $D$ is defined as,
\[
\Pr[D = d] = \frac1{d \cdot (d-1)} \quad \forall d \in \{2,3,4,\ldots\}.
\]
\end{definition}
Observe that, for each positive integer $d$, $\Pr[D > d] = \frac1d$. Recall that in our irregular cuckoo hashing instances the ads are sampled with replacement.\footnote{We point out that our result could also be proved if the sampling of the $D$ ads was done without replacement (with a suitable truncation of the distribution), since the probability that a user chooses the same ad more than once using the model of Definition~\ref{def:instance} --- that is, the probability that the chosen multiset of ads is not a set --- is at most $O\left(n^{-1/2}\right)$.  Indeed, if $S$ is the multiset of the neighbors of a given user, we have that
$\Pr[S \text{ is not a set}] = \sum_{d=2}^{\infty} \left(\Pr[D = d] \cdot \Pr\left[S \text{ is not a set} \mid D = d\right]\right) \le  \sum_{d=2}^{\sqrt{n} } \left(\Pr[D = d] \cdot \Pr\left[S \text{ is not a set} \mid D = d\right]\right) + \sum_{d=\sqrt{n}+1}^{\infty} \Pr[D = d] \le \sum_{d=2}^{\sqrt{n}} \left(\frac1{d\cdot(d-1)} \cdot \binom{d}2 \cdot  \frac n{n^2}\right) + \Pr[D > \sqrt{n}] \le \frac{\sqrt{n}}{2n} + \frac1{\sqrt{n}} = \frac3{2\sqrt{n}}$. Therefore a sampling with replacement  creates, in expectation and with high probability, at most $O\left(\sqrt{n}\right)$ users with parallel edges; removing these users reduces the size of any matching by no more than $O(\sqrt{n})$ edges.}

\smallskip

We will bound the performance of any online algorithm on the instance of Definition~\ref{def:instance} in Section~\ref{sec:LB:online}; we will then show in Section~\ref{sec:LB:maximum} that the instance admits, with high probability, a quasi-complete (and, thus, quasi-perfect) matching.

\subsection{Online Matching}\label{sec:LB:online}
In this Section we upper bound the expected size of the matching found by any online algorithm when presented with the instance of \Cref{def:instance}.

\begin{theorem}\label{thm:om}
    No online algorithm can match more than $$\left(1-\frac{e}{e^e}\right) \cdot n + O\left(\sqrt{n}\right) = 0.8206259\ldots \cdot n + O\left(\sqrt{n}\right)$$ users in expectation with the instance of Definition~\ref{def:instance}. %
\end{theorem}
\begin{proof}
Consider the greedy algorithm that, whenever possible, matches the current user to any of its neighboring ads that are still unmatched. This algorithm is optimal for the instance of \Cref{def:instance}, given that each (multi)set of ads is chosen independently of the others, and that, after conditioning on the (multi)set size $s$, the (multi)set is chosen uniformly at random among those of cardinality $s$.  We then restrict our analysis to this algorithm. 

\smallskip

Let  $S$ be the number of users sampled up until, and including, the point where the greedy algorithm matches $m$ ads, for   $m = \lceil \alpha\cdot n \rceil +1 = \alpha \cdot n + O(1)$, with $\alpha = 1 - e^{1-e}$. Note that it may be that $S\geq n$. Without loss of generality, we can assume that $n$ is larger than any constant, otherwise the statement is trivially true --- in our case, we will take $n\geq 100$.

The crux of the analysis of the  algorithm lies in bounding the expected number of users necessary to achieve a matching of size $m=(1-e^{1-e})n + O(1)$. 
First, we split the value $S$ in $m$ parts.
For $i\in\{0,1,\dots, m-1\}$, let $S_i$ be the number of users sampled in order to match the $(i+1)$-th user, counting from the first user after the $i$-th user is matched (or, if $i = 0$, counting from the first user). We then have $S=\sum_{i=0}^{m-1} S_i$. 
Let $p_j = \frac{1}{j(j-1)}$ for $j\geq 2$. Observe that $S_i \sim \geom(q_i)$, where $q_i = 1- \sum_{j \ge 2} \left(p_j \cdot \left(\frac in\right)^j\right)$. For each $i \in \{0,1,\ldots,m-1\}$, let $x_i = i/n$; then, $0 \le x_i < 1$, and:
\begin{align}
q_i &= 1 - \sum_{j \ge 2} \left(\frac1{j\cdot(j-1)} \cdot  x_i^j\right) = 1 - \sum_{j \ge 2} \left(\left(\frac1{j-1} - \frac1j\right) \cdot  x_i^j\right) \nonumber\\
&=1 - \sum_{j \ge 2} \left(\frac1{j-1}  \cdot  x_i^j\right) + \sum_{j \ge 2} \left(\frac1{j}  \cdot  x_i^j\right) 
=1 - x_i\cdot \sum_{j \ge 1} \left(\frac1{j}  \cdot  x_i^j\right) + \sum_{j \ge 1} \left(\frac1{j}  \cdot  x_i^j\right) - x_i \nonumber \\
&=1 + x_i\cdot \ln(1-x_i) - \ln(1-x_i) - x_i  = (1 - x_i) (1 - \ln(1 - x_i)),\label{eqn:qi}
\end{align}
where we used $\sum_{j \ge 1} \left(\frac1j \cdot x^j\right) = -\ln(1-x)$ for each $x \in [-1,1)$.

\smallskip

We can now bound the expectation of $S$:
\begin{align*}
\frac{\E[S]}{n} & = \frac1n \sum_{i=0}^{m-1} \frac{1}{(1 - x_i)(1-\ln(1-x_i))} \geq \frac1n \sum_{i=1}^{m-1} \frac{1}{(1 - x_i)(1-\ln(1-x_i))} \\
& \geq \int_0^{(m-1)/n} \frac{1}{(1 - x)(1-\ln(1-x))} dx \geq \int_0^{\alpha} \frac{1}{(1 - x)(1-\ln(1-x))} dx \\
& = [\ln(1 - \ln(1 - x))]_{x=0}^{\alpha} = 1,
\end{align*}
where the second inequality follows from the fact that $\frac{1}{(1 - x)(1-\ln(1-x))}$ is increasing in $[0,1)$ and $\frac1n \sum_{i=1}^{m-1} \frac{1}{(1-x_i)(1-\ln(1-x_i))}$ is a right Riemann sum on the interval $[0, (m-1)/n]$. The third inequality holds since $\frac{1}{(1 - x)(1-\ln(1-x))} > 0$ in $[0,1)$ and $(m-1)/n \geq \alpha$. Finally, the solution to the integral can be easily verified. 
Thus, $\E[S] \ge n$. 

\smallskip

We move on to showing the concentration of $S$:
\begin{lemmarep}\label{lem:S-concentrated-om} Let $\delta\in\left[e^{-n/9}, e^{-1}\right]$, and $\epsilon(\delta)=\frac{1}{n}\left\lceil\sqrt{6 n  \ln(\nicefrac{1}{\delta})}\right\rceil$. It holds, $\Pr[S\leq (1-\epsilon(\delta))n]\leq \delta$.
\end{lemmarep}
\begin{proof}

Let $S = \sum_{i=0}^{m-1} S_i$, where $m = 1 + \lceil \alpha\cdot n \rceil$ with $\alpha = 1 - e^{1-e}$, and $S_i\sim \geom(q_i)$, $q_i = (1 - x_i) (1 - \ln(1 - x_i))$, and $x_i = i / n$. Recall that $\E[S] \geq n$. We show that $S$ is concentrated.

The minimum parameter of the geometric variables $S_1, \dots, S_{m-1}$ is $q_{\star} = q_{m-1}$. Applying the standard upper tail bound of \Cref{fact:geom-concentration} we get that, for each $\lambda \in (0,1]$,
\[
\Pr\left[S \le \lambda \cdot \E[S]\right] \le e^{-q_{\star} \cdot \E[S] \cdot (\lambda - 1 - \ln \lambda)}.
\]

Setting $\lambda = 1-\epsilon$ for $\epsilon \in [0,1)$, we then get that
$\Pr\left[S \le (1 - \epsilon) \cdot \E[S]\right] \le e^{-q_{\star} \cdot \E[S] \cdot \frac{\epsilon^2}2}$,
where we used that $-\epsilon - \ln (1-\epsilon) \ge \frac{\epsilon^2}2$ for each $\epsilon \in [0,1)$. Now, set $\epsilon=\epsilon(\delta) = \frac{\left\lceil \sqrt{6\cdot n \cdot \ln \frac{1}{\delta}} \right\rceil}{n} \geq \sqrt{6\cdot \frac{\ln \frac{1}{\delta}}{n}}$. We  have $\delta \geq e^{- \frac{n}{9}}$, therefore, $\epsilon\in [0,1)$. Moreover, it holds,
\begin{align*}
q_{\star} & =\left(1 - \frac{m-1}{n}\right)\left(1 - \ln\left(1-\frac{m-1}{n}\right)\right) =\left(1 - \frac{\lceil \alpha n \rceil}{n}\right)\left(1 - \ln\left(1-\frac{\lceil \alpha n \rceil}{n}\right)\right) \\
& \geq (1-\alpha - 1/n)(1- \ln(1-\alpha - 1/n)) > 1/3,
\end{align*}
where for the first inequality we used that $(1-x)(1 - \ln(1-x))$ is decreasing for $x\in[0,1)$, while the second inequality holds under the hypothesis $n\geq 100$. Thus, we obtain,
\begin{align*}
\Pr[S \leq (1-\epsilon) n] & \leq \Pr[S \leq (1-\epsilon) \E[S]]  \leq e^{-q_{\star} \cdot \E[S] \cdot \frac{\epsilon^2}{2}} \leq \exp\left(-\frac13 \cdot n \cdot \frac{6 \cdot \ln \frac{1}{\delta}}{n} \cdot \frac12 \right)  = \delta,
\end{align*}
where the first inequality follows from $\E[S] \geq n$.
\end{proof}
Now, let $A_k$ be the size of the matching produced by the greedy algorithm  after having processed the first $k$ users (in particular, the matching found by the online algorithm has size $A_n$). It holds $A_k \geq m \iff S \leq k$. Moreover, $A_{n} \leq A_{n-k} + k$. Note also that $\epsilon(\delta)\cdot n \leq 3\sqrt{n\ln(1/\delta)}$ since $\delta \leq e^{-1}$. By \Cref{lem:S-concentrated-om},
\begin{align*}
\Pr[A_n \geq m + 3\sqrt{n\cdot \ln(1/\delta)}] &\leq \Pr[A_n \geq m + \epsilon(\delta)\cdot n] \leq \Pr[A_{(1-\epsilon(\delta))n} \geq m] \\
&= \Pr[S \leq (1-\epsilon(\delta))n] \leq \delta.
\end{align*}

We can finally bound the expected number of matched users,
\begin{align*}
\E[A_n] & = \sum_{k=1}^n \Pr[A_n \geq k]  \leq m + 3\sqrt{n} + \sum_{k=m+\lceil 3\sqrt{n} \rceil }^n \Pr[A_n \geq k] \\ 
& \leq m + 4\sqrt{n} + \sum_{k=3}^{\lceil e^{1-e} \sqrt{n} \rceil} \Pr[A_n \geq m + k\cdot \sqrt{n}] \cdot \sqrt{n} \\
& = m + 4\sqrt{n} +  \sqrt{n} \cdot \sum_{k=3}^{\lceil e^{1-e} \sqrt{n} \rceil} \Pr\left[A_n \geq m + 3\sqrt{n\cdot \ln\frac{1}{e^{-\frac{k^2}{9}}}}\right] \\
& \leq m + 4\sqrt{n} + \sqrt{n} \cdot \sum_{k=3}^{\lceil e^{1-e} \sqrt{n} \rceil} e^{-\frac{k^2}{9}} \le (1 - e^{1-e})n + O(\sqrt{n}),
\end{align*}
where we used the fact %
that $\sum_{k=0}^{\infty} e^{-\frac{k^2}{9}}$ converges to a positive constant.
\end{proof}
A similar argument can be used  to prove that our analysis is tight up to lower order terms.
\begin{lemmarep}\label{lem:analysis-online-u1-tight}
The online greedy algorithm matches at least $\left(1-\frac{e}{e^e}\right)\cdot n - O(\sqrt{n})$ ads in expectation, with the instance of \Cref{def:instance}. %
\end{lemmarep}
\begin{proof}
Let $\alpha=1-e^{1-e}$, $m=\lfloor \alpha n \rfloor$. Let $x_i=\frac{i}{n}$, $q_i=(1-x_i)(1-\ln(1-x_i))$ for $i\in\{0,\dots, m-1\}$. Without loss of generality, we assume $n\geq 100$. Note that $x_i\in[0,1)$. Suppose we keep sampling users (regardless of the value of $n$) until $m$ of them get matched; let $S$ be the total number of  sampled users. We have $q_{m-1}\geq (1-\alpha)(1-\ln(1-\alpha))>\frac{1}{3}$,
 and,
\begin{align*}
\frac{\E[S]}{n} &= \frac{1}{n} \sum_{i=0}^{m-1} \frac{1}{(1-x_i)(1-\ln(1-x_i))}  \leq \int_0^{m/n} \frac{1}{(1-x)(1-\ln(1-x))} \, dx \\
& \leq \int_0^{\alpha} \frac{1}{(1-x)(1-\ln(1-x))} \, dx = 1.
\end{align*}
Thus, $\E[S]\leq n$. Similarly, by lower bounding the sum with an integral, we also have $\E[S]\geq \alpha n$. For $\delta\in[e^{-\frac{n\alpha^3}{12}},e^{-\alpha/12}]$, let $\epsilon=\epsilon(\delta)=\sqrt{\frac{12}{\alpha n}\cdot \ln(\frac{1}{\delta})} \in(0,1)$. By \Cref{fact:geom-concentration}, and since $\epsilon-\ln(1+\epsilon) \geq \frac{\epsilon^2}{4}$ for $\epsilon\in(0,1)$, we have,
\begin{align*}
\Pr[S\geq (1+\epsilon)\E[S]] \leq e^{-q_{m-1}\cdot \E[S] \cdot \frac{\epsilon^2}{4}} \leq \delta.
\end{align*}
Let $A_k$ be the number of ads matched after sampling $k$ users. We have,
\begin{align*}
\Pr[A_n \geq m - \sqrt{\nicefrac{12}{\alpha}\cdot n \cdot \ln(\nicefrac{1}{\delta})}] &= \Pr[A_n \geq m - \epsilon n] \geq \Pr[A_{\lfloor (1+\epsilon)n \rfloor} \geq m]\\
&= \Pr[S\leq (1+\epsilon)n] \geq \Pr[S\leq (1+\epsilon)\E[S]]  \\
&\geq 1 - \delta.
\end{align*}
We can finally compute the expectation of $\E[A_n]$. 
\begin{align*}
\E[A_n] & = \sum_{k=1}^n \Pr[A_n \geq k] \geq \sum_{k=1}^{\lfloor m-\sqrt{n} \rfloor} \Pr[A_n \geq k] \geq \sum_{k=1}^{\left\lfloor \frac{\lfloor m - \sqrt{n} \rfloor}{\sqrt{n}+1} \right\rfloor} (\sqrt{n}-1)  \cdot \Pr[A_n \geq m - k\sqrt{n}]\\
& = \sum_{k=1}^{\left\lfloor \frac{\lfloor m - \sqrt{n} \rfloor}{\sqrt{n}+1} \right\rfloor} (\sqrt{n}-1)  \cdot \Pr\left[A_n \geq m - \sqrt{\frac{12}{\alpha} \cdot n\cdot \ln\left(\frac{1}{e^{-k^2\alpha/12}}\right)}\right]\\
& \geq \sum_{k=1}^{\left\lfloor \frac{\lfloor m - \sqrt{n} \rfloor}{\sqrt{n}+1} \right\rfloor} (\sqrt{n}-1) (1 - e^{-k^2\alpha/12}) \\
&= (\sqrt{n}-1)\left\lfloor \frac{\lfloor m - \sqrt{n} \rfloor}{\sqrt{n}+1} \right\rfloor - (\sqrt{n}-1)\sum_{k=1}^{\left\lfloor \frac{\lfloor m - \sqrt{n} \rfloor}{\sqrt{n}+1} \right\rfloor} e^{-k^2\alpha/12}\\
& = m - \Theta(\sqrt{n})
\end{align*}
where we used that $\sum_{k=1}^\infty e^{-k^2\alpha/12}$ converges to a positive constant.
\end{proof}

\subsection{Maximum Matching}\label{sec:LB:maximum}
We now prove that the instance of \Cref{def:instance} admits a matching of $(1-o(1)) \cdot n$ edges. 

We will analyze our instance using the Karp-Sipser algorithm; %
first, we will reduce our instance to a configuration model; then, we will study this configuration model with a modification of the analysis in \cite{bg15} of the first phase of the Karp-Sipser algorithm. In order to implement this plan, we need to modify our instance so that its average degree is bounded.
\begin{definition}\label{def:modified-instance}
Given any integer $\Delta \geq 2$, we consider the irregular cuckoo hashing distribution $\irregularcuckoohashing(n, n, D_{\Delta})$, where $D_\Delta$ is defined as
\begin{align*}
\Pr[D_\Delta = 0]  = \frac{1}{\Delta}, &\quad \text{ and } \quad
\Pr[D_\Delta = i]  = \frac{1}{i\cdot(i-1)} \quad \forall i \in \{2,\dots, \Delta\}.
\end{align*}
\end{definition}
Observe that, to sample an instance of \Cref{def:modified-instance}, one can first sample an instance of \Cref{def:instance} and then remove all edges incident on users of degree larger than $\Delta$. We will prove the following result in Section~\ref{sec:red_conf_mod}.
\begin{lemma}\label{lem:max-matching-cut-distr}
With probability $1-o(1)$, the instance of \Cref{def:modified-instance} admits a matching of size $\left(1 - \frac{1}{\Delta}\right)n - o(n)$.
\end{lemma}
The existence of a quasi-complete matching in our original instance can be easily derived from Lemma~\ref{lem:max-matching-cut-distr}. %
\begin{theoremrep}\label{thm:mm}
With probability $1-o(1)$, the instance of \Cref{def:instance} admits a matching of size $(1-o(1))n$.
\end{theoremrep}
\begin{proof}
Consider the instance $G(n)$ of~\Cref{def:instance} with $n$ users. Let $G_{\Delta}(n)$ be a copy of $G(n)$ where all the edges incident on users whose degree is larger than $\Delta$ in $G(n)$ are removed. Then $G_{\Delta}(n)$ is distributed as the instance of~\Cref{def:modified-instance}. Let $M^{\star}(n)$ (resp. $M_{\Delta}^{\star}(n)$) be the size of the maximum matching in $G(n)$ (resp., $G_{\Delta}(n)$). Note that $M^{\star}(n) \geq M_{\Delta}^{\star}(n)$ for any $\Delta$. %
For any $\epsilon \in (0,1)$, by selecting $\Delta = \left\lceil 2/\epsilon \right\rceil$, we have that, by \Cref{lem:max-matching-cut-distr} and by using $M^{\star}(n) \geq M_{\Delta}^{\star}(n)$, there exists $n_\epsilon$ such that for all $n \geq n_\epsilon$, $\Pr\left[M^{\star}(n) \geq (1-\epsilon)n\right] \geq 1 - \epsilon$. Therefore, by definition, $M^{\star}(n) \geq (1-o(1))n$ with probability $1-o(1)$.
\end{proof}
We also note that our main Theorem follows from Theorems~\ref{thm:om} and~\ref{thm:mm}:
\begin{theorem}\label{thm:mainlb}
The optimal competitive ratio %
for the online matching problem with irregular cuckoo hashing  distributions %
is no better than $1 - \frac e{e^e} + o(1) \approx 0.8206259\ldots$.
\end{theorem}
Clearly, Theorem~\ref{thm:mainlb} directly applies to the cases of ``IID users with known distribution'', ``IID users with unknown distribution'' and ``adversarial graphs whose users are permuted uniformly at random''.

\smallskip

Moreover, given that Theorem~\ref{thm:mm} guarantees that the maximum matching has size $(1-o(1)) \cdot n$ with probability at least $1-o(1)$, our upper bound of $1-\frac{e}{e^e} + o(1)$ holds for both the ratios $\frac{\E[\text{ALG}]}{\E[\text{OPT}]}$ and $\E\left[\frac{\text{ALG}}{\text{OPT}}\right]$,\footnote{%
If $\text{OPT} = 0$, the generic online algorithm  necessarily returns the maximum (empty) matching;  it is typical to define $\frac{\text{ALG}}{\text{OPT}} = 1$ in this borderline case. Moreover, if $\E[\text{OPT}] = 0$ then necessarily $\Pr[\text{OPT} = 0] = 1$,  so that it is natural to define $\frac{\E[\text{ALG}]}{\E[\text{OPT}]} = 1$ in this other case, as well.}  that is, for both the ratio-of-the-expectations and the expectation-of-the-ratio definitions of competitive ratio.

\subsubsection{Reduction to the Configuration Model}\label{sec:red_conf_mod}
In this Subsection, we aim to prove~\Cref{lem:max-matching-cut-distr}.
We start by showing that, if we condition on the degrees of the vertices, then our instance is distributed like a configuration model.
The following Lemma is the analogue of many that have been proved for various random graph models; a famous example is that of random regular graphs~\cite{jlr11}.
\begin{lemmarep}\label{lem:reduction-conf-model}%
For any non-negative integers $\mathbf{d}=(d_1, d_2, \dots, d_n),$ $\mathbf{\hat{d}}=(\hat{d}_1, \hat{d}_2, \dots, \hat{d}_n)$ such that $\sum_{i=1}^n d_i = \sum_{i=1}^n\hat{d}_i$, and for any bipartite graph $G$, it holds,
\[
\Pr_{X \sim \irregularcuckoohashing(n,n,P)}\left[X = G \ \middle| \ \deg(v_i)=d_i\wedge \deg(\hat{v}_i)=\hat{d}_i\text{ for each $i\in[n]$}\right] = \Pr_{X \sim \CM(\mathbf{d}, \mathbf{\hat{d}})}[X = G],
\]
where $P$ is such that $\Pr[P=d_i]>0$ for each $i\in[n]$, so that the event $\{\deg(v_i)=d_i\wedge\deg(\hat{v}_i)=\hat{d}_i\text{ for each $i\in[n]$}\}$ happens with positive probability. 
\end{lemmarep}
\begin{appendixproof}
Let us denote the bipartite graph $G$ with an adjacency matrix $(g_{ij})_{i,j\in[n]}$, where $g_{ij}$ is the number of edges between $v_i\in V$ and $\hat{v}_j \in \hat{V}$. For simplicity, let us use $L$ in place of $\irregularcuckoohashing(n,n,P)$. Both probabilities are zero unless $\sum_{j\in[n]}g_{ij}=d_i$ for each $i$, and $\sum_{i\in[n]} g_{ij} = \hat{d}_j$ for each $j$. Otherwise, if both equations hold, we have: 
\begin{align*}
\Pr_{X \sim L}\left[X = G \ \middle| \ \deg(v_i)=d_i \wedge \deg(\hat{v}_i)=\hat{d}_i\ \forall i\in[n]\right] & = \frac{\Pr_{X \sim L}\left[X = G \right]}{\Pr_{X \sim L}\left[\deg(v_i)=d_i \wedge \deg(\hat{v}_i)=\hat{d}_i \ \forall i\in[n]\right]}
\end{align*}
Let $p_i = \Pr_{X\sim L}[\deg(v_i)=d_i]=\Pr[P=d_i]$. Conditioning on the degrees of $V$, the degrees of $\hat{V}$ follow a multinomial distribution, therefore:
\begin{align*}
\Pr_{X \sim L}\left[\deg(v_i)=d_i \wedge \deg(\hat{v}_i)=\hat{d}_i \ \forall i\right] & = \Pr_{X \sim L}\left[\deg(v_i)=d_i\ \forall i\right] \cdot \Pr_{X \sim L}\left[\deg(\hat{v}_i)=\hat{d}_i \ \forall i \ \middle|\ \deg(v_i)=d_i\ \forall i\right] \\
& = \left(\prod_{i\in[n]} p_{i}\right)  \cdot \frac{\left(\sum_{i\in[n]} d_i \right)!}{\prod_{i\in[n]}n^{\hat{d}_i} \cdot \hat{d}_i!}.
\end{align*}
Since the neighborhood of each $v_i$ is sampled independently, we have:
\begin{align*}
\Pr_{X\sim L}[X = G] & = \Pr_{X \sim L}[x_{ij} = g_{ij}\ \forall i,j] = \prod_{i\in[n]} \Pr_{X \sim L}[x_{ij} = g_{ij}\ \forall j] \\
 & = \prod_{i\in[n]} p_{i} \cdot \Pr_{X \sim L}\left[x_{ij} = g_{ij}\ \forall j \mid \deg(v_i)=d_i\right] \\ 
 & = \prod_{i\in[n]} \left( p_{i} \cdot \frac{d_i!}{n^{d_i} \prod_{j\in[n]}g_{ij}!} \right)
\end{align*}
where the last equality follows by using $\sum_{j\in[n]}g_{ij}=d_i$ and by the multinomial distribution. Using $\sum_{i\in[n]} d_i = \sum_{i\in[n]} \hat{d}_i$ and simplifying, we obtain:
\[
\Pr_{X \sim L}\left[X = G \ \middle| \ \deg(v_i)=d_i \wedge \deg(\hat{v}_i)=\hat{d}_i\ \forall i\in[n]\right] = \frac{\prod_{i\in[n]} d_i! \hat{d}_i!}{\left(\sum_{i\in[n]} d_i\right)! \cdot \prod_{i\in[n]}\prod_{j\in[n]}g_{ij}!}.
\]

Consider now the configuration model. There are $\left(\sum_{i\in[n]}d_i\right)!$ possible matchings between the configuration points, each having the same probability. Let us count the number of matchings that result into graph $G$. We have $\binom{\hat{d}_1}{g_{11}}$ ways to select the configuration points from $\hat{v}_1$ to be attached to $v_1$, and similarly for $\hat{v}_2, \dots, \hat{v}_n$. Once chosen the configuration points to be matched with $v_1$, we have $d_1!$ ways to permute them. Then, there are $\binom{\hat{d}_1 - g_{11}}{g_{21}}$ ways to select the configuration points from $\hat{v}_1$ to be attached to $v_2$, and so on. Repeating this process we can see that the number of good matchings is:
\begin{align*}
\prod_{i\in[n]} \left( d_i! \prod_{j\in[n]} \binom{\hat{d}_j - \sum_{k=1}^{i-1}g_{kj}}{g_{ij}} \right) & = \frac{\prod_{i\in[n]}d_i!}{\prod_{i\in[n]}\prod_{j\in[n]} g_{ij}!} \cdot \prod_{i\in[n]}\prod_{\substack{j\in[n] :\\g_{ij}>0}} \left(\hat{d}_j - \sum_{k=1}^{i-1} g_{kj}\right) \dots \left(\hat{d}_j - \sum_{k=1}^{i} g_{kj} + 1\right) \\
& = \frac{\prod_{i\in[n]}d_i!}{\prod_{i\in[n]}\prod_{j\in[n]} g_{ij}!} \cdot \prod_{j\in[n]}\prod_{i\in[\hat{d}_j]} (\hat{d}_j - i + 1)\\
& = \frac{\prod_{i\in[n]} d_i! \hat{d}_i!}{\prod_{i\in[n]}\prod_{j\in[n]} g_{ij}!}
\end{align*}
where all the binomial coefficients are well defined since $\sum_{i\in[n]} g_{ij} = \hat{d}_j$. This concludes the proof.
\end{appendixproof}
We now prove that, in our instance, the number of vertices with a certain degree is concentrated around the mean. This will allow us to restrict our attention to a specific configuration model. Define, for $i\geq 0$,
\begin{align}
&z_i = \begin{cases} \frac{1}{\Delta} & \text{if $i=0$} \\ 0 & \text{if $i=1$ or $i>\Delta$} \\ \frac{1}{i(i-1)} & \text{if $2 \leq i \leq \Delta$} \end{cases}  \quad & \hat{z}_i & =\frac{H_{\Delta-1}^i e^{-H_{\Delta-1}}}{i!}. \label{def:zi-main-instance}
\end{align}
The following Observation can be proved by applying the concentration bounds in \Cref{sec:preliminaries}.
\begin{observationrep}\label{obs:graph-concentrated}
Consider an irregular cuckoo hashing instance $\irregularcuckoohashing(n,n,P)$ such that $P$ is upper bounded by a constant $\Delta>0$. Let $p_i=\Pr[P=i]$ for $0\leq i\leq \Delta$, $\mu = \E[\deg(v)] = \sum_{i=0}^\Delta i p_i$, and let $\hat{p_i}=\frac{e^{-\mu} \cdot \mu^i}{i!}$, $i\geq 0$, be a Poisson distribution with mean $\mu$. We have that, with probability $1-o(1)$, 
\begin{enumerate}
    \item $|E|=(\mu \pm o(n^{-\nicefrac15})) n$,
    \item $n_i=(p_i \pm o(n^{-\nicefrac15}))n$ for $0 \leq i \leq \Delta$,
    \item $\hat{n}_i = (\hat{p}_i \pm o(n^{-\nicefrac15}))n$ for $i\geq 0$.
\end{enumerate}
\end{observationrep}
\begin{appendixproof}
Note $|E|=\sum_{i\in[n]} \deg(v_i)$, then, $E\left[|E|\right] = \sum_{i=1}^n \sum_{d=0}^{\Delta} d\cdot p_d = n\cdot \mu$. Since $\deg(v_i)\in[0,\Delta]$, applying \Cref{fact:chernoff-hoeffding} we get:
\[
\Pr\left[\middle| |E| - \E[|E|] \middle| \geq n^{3/4}\right] \leq 2 \exp\left(-\frac{2 \sqrt{n}}{\Delta^2}\right) = o(1).
\]
Similarly, letting $X_i^{(d)}=\begin{cases} 1 & \text{if $\deg(v_i)=d$} \\ 0 & \text{otherwise} \end{cases}$, we have $n_{d} = \sum_{i=1}^n X_i^{(d)}$, and $\E[n_{d}] = n\cdot p_d$, then:
\[
\Pr\left[\exists\ 0\leq d \leq \Delta \text{ s.t.: } \middle| n_d - \E[n_d] \middle| \geq n^{3/4} \right] \leq \sum_{d=0}^{\Delta} \Pr\left[\middle| n_d - \E[n_d] \middle| \geq n^{3/4} \right] \leq 2(\Delta+1)\cdot e^{-2\sqrt{n}} = o(1).
\]
To prove the third statement, condition on $|E|=(\mu \pm o(1))n$ that, as we proved, happens with probability $1-o(1)$. For $i\in[|E|]$, let $Y_i \in [n]$ denote the vertex in $\hat{V}$ selected with the $i$-th edge. For $0\leq d \leq |E|$, $\hat{n}_d$ can be expressed as a function $f_d$ of the $Y_i$'s: 
\[
f_d(Y_1, Y_2, \dots, Y_{|E|}) = \sum_{i=1}^n \left[\left(\sum_{j=1}^{|E|}[Y_j = i]\right) = d\right],
\] 
where $[P]$ is the Iverson bracket for a predicate $P$. Note that the function $f_d$ changes by at most 2 if applied on two inputs differing only in the $i$-th coordinate, that is, for all $1\leq i \leq |E|$ and $y_1, \dots, y_{|E|}, y'_i$:
\[
| f_d(y_1, \dots, y_{i-1}, y_i, y_{i+1}, \dots, y_{|E|}) - f_d(y_1, \dots, y_{i-1}, y'_i, y_{i+1}, \dots, y_{|E|})| \leq 2.
\]
Then, given that the $Y_i$'s are independent, we can apply \Cref{fact:mcdiarmid}:
\begin{align*}
\Pr\left[ \exists\ 0\leq d \leq |E| \text{ s.t.: } \middle| \hat{n}_d - \E[\hat{n}_d] \middle| \geq n^{3/4} \right] & \leq \sum_{d=0}^{|E|} \Pr\left[ \middle| \hat{n}_d - \E[\hat{n}_d] \middle| \geq n^{3/4} \right] \\
& \leq 2 (|E|+1) \exp\left(-\frac{\sqrt{n}}{2\cdot \mu \pm o(1)}\right)\\
& = o(1),
\end{align*}
where in the last two steps we use that $|E|=(\mu \pm o(1))n$. 

Now, observe that for a vertex $\hat{v}_i \in \hat{V}$, $\deg(\hat{v}_i) \sim \Bin(|E|, 1/n)$, therefore $\E[\hat{n}_d] = \sum_{i=1}^n \Pr[\deg(\hat{v}_i) = d]$. By \Cref{fact:lecam}, we have, 
\[
\left|\E[\hat{n}_d] - n \hat{p}_d\right| \leq n\cdot \dTV(\Bin(|E|,1/n), \Poisson(\mu)) \leq \mu + 2n^{3/4}, 
\]
where we used that $|E|=(\mu \pm o(1))n$, and in particular, $n\mu - n^{3/4} \leq |E|\leq n\mu + n^{3/4}$. By triangle inequality, we get,
\[
\Pr\left[ \exists\ 0\leq d \leq |E| \text{ s.t.: } \middle| \hat{n}_d - n\hat{p}_d \middle| \geq 3n^{3/4} + \mu \right] \leq o(1). \qedhere
\]
\end{appendixproof}
Note that \Cref{obs:graph-concentrated} applies to our graph with $p_i=z_i$ and $\hat{p_i}=\hat{z_i}$. Therefore, we can now leverage on Theorem~\ref{thm:bg}, that is, a variant of the result in~\cite{bg15}, which uses the Karp-Sipser algorithm to bound the maximum matching in  configuration models whose degrees follow the conditions of \Cref{obs:graph-concentrated}. Note that Theorem~\ref{thm:bg} requires a bounded average degree; we made sure that our random graph has bounded average degree by limiting its maximum user degree. Formally, 
\begin{lemma}\label{lem:karp-sipser-conf-model}
For any constant $\Delta\geq 2$, let $\mathbf{d}\in\mathbb{Z}_{\geq0}^n$ and $\mathbf{\hat{d}}\in\mathbb{Z}_{\geq0}^n$ be any two valid degree sequences such that (i) $|E|=(H_{\Delta-1}\pm o(n^{\nicefrac{-1}5}))n$, (ii) $n_i=(z_i \pm o(n^{\nicefrac{-1}5}))n$, and (iii) $\hat{n}_i = (\hat{z}_i \pm o(n^{\nicefrac{-1}5}))n$ for $i\geq 0$, where $z_i$ and $\hat{z}_i$ are defined in \Cref{def:zi-main-instance}. Then, with probability $1-o(1)$, $\CM(\mathbf{d}, \mathbf{\hat{d}})$ has a matching of size at least $(1-1/\Delta -o(1))n$.
\end{lemma}
\begin{proof}
We apply \Cref{thm:bg} with $z_i$ and $\hat{z}_i$ as defined in \Cref{def:zi-main-instance}. Let $M=\ln(n)$, and without loss of generality, take $n>e^{\Delta+2}=O(1)$, so that $M>\Delta+2\ge H_{\Delta-1}+3$. We have, $\sum_{i=M}^\infty i^2 \hat{z}_i \leq \sum_{i=M}^{\infty} \frac{e^{-H_{\Delta-1}}(H_{\Delta-1})^{i}}{(i-3)!} = (H_{\Delta-1})^3\sum_{i=M-3}^{\infty} \frac{e^{-H_{\Delta-1}}(H_{\Delta-1})^{i}}{i!} = (H_{\Delta-1})^3 \Pr[\Poisson(H_{\Delta-1})\geq M - 3] \leq \frac{(H_{\Delta-1})^{M} \cdot e^{M-3-H_{\Delta-1}}}{(M-3)^{M-3}}  = \frac{2^{\Theta(\log n)}}{\Theta(\log n)^{\Theta(\log n)}} \leq o(n^{-\nicefrac15})$, where we used the standard tail bound $\Pr[\Poisson(\lambda)\geq x] \leq \frac{(e\lambda)^x e^{-\lambda}}{x^x}$ for $x>\lambda$ (see, e.g., \cite[Theorem 5.4]{mu05}).
We have,
\begin{align*}
    f(x) & = \frac{1}{\Delta} + \sum_{i=2}^\Delta \left(\frac{1}{i(i-1)}\cdot x^i\right)\\
    f'(x) &=\sum_{i=1}^{\Delta-1}\frac{x^i}{i} = \begin{cases}
    H_{\Delta-1} & \text{if } x=1\\
    -\ln(1-x) - \sum_{i=\Delta}^\infty \frac{x^i}{i} & \text{if } |x|<1
    \end{cases}\\
    \hat{f}(x)&= \sum_{i=0}^\infty \frac{H^i_{\Delta-1 }e^{-H_{\Delta-1}}}{i!} \cdot x^i = e^{(x-1)H_{\Delta-1}}\cdot \sum_{i=0}^\infty \frac{(x \cdot H_{\Delta-1})^i e^{-x\cdot H_{\Delta-1}}}{i!} = e^{(x-1)\cdot H_{\Delta-1}} \\
    \hat{f}'(x)&= H_{\Delta-1} \cdot e^{(x-1)\cdot H_{\Delta-1}},
\end{align*}
where we used $\sum_{i=0}^\infty \frac{\lambda^i e^{-\lambda}}{i!} = 1$ for all $\lambda$ and $\sum_{i=1}^\infty \frac{x^i}{i} = -\ln(1-x)$ for all $|x|<1$. Note that $f'(1)=\hat{f}'(1)=H_{\Delta-1} < \infty$, $\hat{f}'(2) = H_{\Delta-1} e^{H_{\Delta-1}} = O(1)$, and clearly $f'(2)=\sum_{i=1}^\Delta z_i \cdot i \cdot 2^{i-1} = O(1)$. Moreover, $|V|=|\hat{V}|=n=f(1)n=\hat{f}(1)n$. Therefore, the hypotheses of \Cref{thm:bg} are met, we now compute $\hat{w}_1$ and $w_2$ to apply the Theorem. Recall that $\hat{w}_1 = \frac{\hat{f}'(w_2)}{H_{\Delta-1}} = e^{(w_2-1) \cdot H_{\Delta-1}}$. Since $w_2\in[0,1]$, we have that $\hat{w}_1 > 0$, and therefore $1 - \hat{w}_1 \in [0,1)$. Substituting this expression for $\hat{w}_1$ into the relation $w_2=1-\frac{f'(1-\hat{w}_1)}{H_{\Delta-1}}$, we get,
\begin{align*}
    w_2 &= 1 - \frac{-\ln(e^{(w_2-1)H_{\Delta-1}}) - \sum_{i=\Delta}^\infty \frac{\left(1-e^{(w_2-1)H_{\Delta-1}}\right)^i}{i}}{H_{\Delta-1}} = w_2 + \sum_{i=\Delta}^\infty \frac{\left(1 - e^{(w_2-1)H_{\Delta-1}}\right)^i}{i\cdot H_{\Delta-1}}
\end{align*}
Equivalently,
$
\sum_{i=\Delta}^\infty \frac{\left(1 - e^{(w_2-1)H_{\Delta-1}}\right)^i}{i} = 0
$
and this equation is satisfied only for $w_2 = 1$. Indeed, for $x\in[0,1)$, $\frac{(1-e^{(x-1)H_{\Delta-1}})^\Delta}{\Delta} > 0$. Therefore, $w_2 = \hat{w}_1=1$. Now, \Cref{thm:bg} ensures that the configuration model admits a matching of size at least,
\[
(f(1) - f(0) - f'(0) - o(1))n = (1 - 1/\Delta - o(1))n.  \qedhere
\]
\end{proof}
We now apply \Cref{lem:reduction-conf-model}, \Cref{obs:graph-concentrated} and \Cref{lem:karp-sipser-conf-model}, to prove \Cref{lem:max-matching-cut-distr}, and conclude the proof of our impossibility result.
\begin{proof}[Proof of \Cref{lem:max-matching-cut-distr}]
\newcommand{\good}{\text{good}}
By \Cref{obs:graph-concentrated}, there exists $h(n)=o(n^{\nicefrac45})$ such that, in the instance of \Cref{def:modified-instance}, with probability $1-o(1)$, (i) $|n_i - nz_i| \leq h(n)$, (ii) $|\hat{n}_i - n\hat{z}_i| \leq h(n)$ for all $i\geq 0$, and (iii) $| |E| - n H_{\Delta-1}| \leq h(n)$. Let $D_n$ be the set of all sequences $\mathbf{d}\in\mathbb{Z}_{\geq0}^n$, $\mathbf{\hat{d}}\in\mathbb{Z}_{\geq0}^n$ that are compatible with conditions (i), (ii), and (iii). Note that $D_n\neq\varnothing$ given that with probability $1-o(1)$, the sampled graph respects the conditions. Let $L$ be the distribution defined in \Cref{def:modified-instance}. For a graph $G$, let $M(G)$ be the event that the maximum matching of $G$ has size at least $(1-1/\Delta-o(1))n$ (where the $o(1)$ term is the one given by \Cref{lem:karp-sipser-conf-model}), let $\good(G)$ be the event that $G$ respects conditions (i),(ii), and (iii), and let $D(G, \mathbf{d}, \mathbf{\hat{d}})$ be the event that the vertices of $G$ have degrees $\mathbf{d}$ and $\mathbf{\hat{d}}$. Note that $\{D(G,\mathbf{d}, \mathbf{\hat{d}})\}_{(\mathbf{d}, \mathbf{\hat{d}})\in D_n}$ is a partition of $\good(G)$. We have,
\begin{align*}
\Pr_{G\sim L}\left[M(G)\right] & \geq \Pr_{G\sim L}[\good(G)]\Pr_{G\sim L}[M(G) \mid \good(G)] \\
& \geq (1-o(1)) \sum_{(\mathbf{d}, \mathbf{\hat{d}})\in D_n} \Pr_{G\sim L}[M(G) \mid D(G,\mathbf{d}, \mathbf{\hat{d}})] \Pr_{G\sim L}[D(G,\mathbf{d}, \mathbf{\hat{d}})\mid \good(G)]\\
& = (1-o(1))\sum_{(\mathbf{d}, \mathbf{\hat{d}})\in D_n} \Pr_{G\sim \CM(\mathbf{d}, \mathbf{\hat{d}})}[M(G)] \Pr_{G\sim L}[D(G,\mathbf{d}, \mathbf{\hat{d}})\mid \good(G)] \\
& \geq (1-o(1))^2 \sum_{(\mathbf{d}, \mathbf{\hat{d}})\in D_n} \Pr_{G\sim L}[D(G,\mathbf{d}, \mathbf{\hat{d}})\mid \good(G)] = 1-o(1),
\end{align*}
where in the first equality we use \Cref{lem:reduction-conf-model} and in the last inequality we apply \Cref{lem:karp-sipser-conf-model}.
\end{proof}

\section{Extremality}\label{sec:tightness}

We prove in this section that the construction of \Cref{def:instance} is extremal among the irregular cuckoo hashing instances in the following sense: %
for each $d = 0, 1, 2, \ldots$, our distribution (greedily) maximizes the total number of users of degree at most $d$ while guaranteeing that a $1-o(1)$ fraction of these users can be matched to ads, %
that is, while guaranteeing that the resulting graph admits a quasi-complete matching. The distribution, then, aims to make the user degrees as small as possible (so to make the task of the  online algorithm as hard  as possible), while guaranteeing that  offline algorithms can match a $1-o(1)$ fraction of the users.

\smallskip

We will prove that, for each fixed $\Delta \ge 2$, if we increase the probability that Definition~\ref{def:modified-instance} assigns to user-degree $\Delta$  by any small enough positive constant $\epsilon > 0$, while keeping the probabilities of user-degrees $1,\ldots,\Delta-1$ fixed, then a constant fraction of the users having degree in $\{1,2,\ldots,\Delta\}$ cannot be matched. In particular, we will consider the following class of distributions.
\begin{definition}\label{def:eps-mass-instance}
Choose an integer constant $\Delta\geq 2$ and a small enough constant $\epsilon>0$. Consider the irregular cuckoo hashing instance $\irregularcuckoohashing(n,n,D_{\Delta,\epsilon})$, where $D_{\Delta,\epsilon}$ is defined as,
\begin{align*}
\Pr[D_{\Delta,\epsilon} = 0] &= \Pr[D_\Delta=0] - \epsilon\\
\Pr[D_{\Delta,\epsilon} = i] &= \Pr[D_\Delta=i] \quad\quad\quad \forall i \in \{1,2,\ldots,\Delta-1\}\\ 
\Pr[D_{\Delta,\epsilon} = \Delta] &= \Pr[D_\Delta=\Delta] + \epsilon,
\end{align*}
where $D_{\Delta}$ is given in \Cref{def:modified-instance}.
\end{definition}

We will again analyze the first phase of the Karp-Sipser algorithm by using \Cref{thm:bg}. 

\begin{lemmarep}\label{lem:eps-mass-extremality}
With probability $1-o(1)$, the instance of \Cref{def:eps-mass-instance} for any constant $\Delta\ge 2$ and for a small enough constant $\epsilon>0$ with respect to $\Delta$, does not admit matchings of size larger than
$
\left(1 - \frac{1}{\Delta}+\epsilon - \psi_\Delta(\epsilon) +o(1)\right)\cdot n
$
where $\psi_\Delta(\epsilon)>0$ for $\epsilon$ sufficiently small compared to $\Delta$. Specifically, $\psi_\Delta(\epsilon)=\frac{\Delta^{2\Delta}}{\Delta+1}\cdot \epsilon^{\Delta+1} \pm O\left(c_{\Delta} \cdot \epsilon^{\Delta+2}\right)$, where $c_{\Delta}$ depends only on $\Delta$. In particular,  with probability $1-o(1)$, each matching will leave  at least $(\psi_\Delta(\epsilon)-o(1)) \cdot n$ users of %
positive degree unmatched.%
\end{lemmarep}
\begin{appendixproof}
With a similar argument as the one used in \Cref{lem:max-matching-cut-distr}, we can reduce to studying a configuration model where (i) $|E|=(H_{\Delta-1} + \epsilon\Delta \pm o(n^{\nicefrac{-1}5}))n$, (ii) $n_0=(1/\Delta-\epsilon \pm o(n^{\nicefrac{-1}5}))n$, $n_1=0$, $n_i=(\frac{1}{i(i-1)} \pm o(n^{\nicefrac{-1}5}))n$ for $2\leq i\leq \Delta-1$, $n_{\Delta}=(\frac{1}{\Delta(\Delta-1)}+\epsilon \pm o(n^{\nicefrac{-1}{5}}))n$, and (iii) $\hat{n}_i=\left(\frac{(H_{\Delta-1} + \epsilon\Delta)^i e^{-H_{\Delta-1} - \epsilon\Delta}}{i!} \pm o(n^{\nicefrac{-1}5})\right)n$ for $i\ge 0$. We aim to upper bound the maximum matching in this configuration model via \Cref{thm:bg}. We have:
\begin{align*}
f(x) &=  \left(\frac{1}{\Delta} - \epsilon\right) +\sum_{j = 2}^{\Delta} \left(\frac1{j\cdot(j-1)} \cdot x^j\right) + \epsilon \cdot x^\Delta\\
&=\left(\frac1{\Delta}-\epsilon\right)+\left(1-x\right)\cdot \left(\ln(1-x)+x^{\Delta+1}\cdot\Phi(x,1,\Delta+1)\right) + x \cdot \left(1- \frac{x^{\Delta}}{\Delta}\right) + \epsilon \cdot x^{\Delta},\\
f'(x)&= \sum_{j = 2}^{\Delta} \left(\frac1{j-1}  x^{j-1}\right) + \Delta \epsilon  x^{\Delta-1}= \Delta \epsilon  x^{\Delta-1}+\left\{\begin{array}{ll}
-\ln(1-x)  -  \sum_{j=\Delta}^{\infty} \frac{x^j}j & \text{if } |x| < 1\\
H_{\Delta-1} & \text{if } x = 1\end{array}\right.\\
&=\Delta \epsilon  x^{\Delta-1}+\left\{\begin{array}{ll}
-\ln(1-x)  - x^{\Delta}\cdot \Phi(x,1,\Delta) & \text{if } |x| < 1\\
H_{\Delta-1} & \text{if } x = 1\end{array}\right. \\
\hat{f}(x) & = e^{-(1-x)(H_{\Delta-1} + \epsilon\Delta)}\\
\hat{f}'(x) & = (H_{\Delta-1}+\epsilon\Delta) \cdot e^{-(1-x)(H_{\Delta-1} + \epsilon\Delta)}\\
\mu & = f'(1)=\hat{f}'(1) = H_{\Delta-1} + \epsilon\Delta
\end{align*}%
where we made use of the Taylor expansion $-\ln(1-x)=\sum_{i=1}^\infty \frac{x^i}{i}$ and of the Lerch transcendent $\Phi(z, s, \alpha)=\sum_{i=0}^\infty \frac{z^i}{(i+\alpha)^s}$. We now compute $\hat{w}_1$ and $w_2$, and then apply \Cref{thm:bg}. Let $y$ and $z$ be solutions of $z=1 - \frac{f'(1-y)}{\mu}$, $y=\frac{\hat{f}'(z)}{\mu}$. Note that $(y=1,z=1)$ is the only solution where $y=1$ or $z=1$. Moreover, $(y=0, z)$ and $(y, z=0)$ are not valid solutions. Since we are interested in the smallest possible solutions, let us assume $y,z \in (0,1)$. By substituting $y=e^{-(1-z)\mu}$ into the definition of $z$, we have,
\begin{align*}
\mu \cdot z & = \mu - f'(1-y) = \mu - \Delta\epsilon(1-y)^{\Delta-1} + \ln(y) + \sum_{j=\Delta}^\infty \frac{(1-y)^j}j\\
& = \mu - \Delta\epsilon(1-y)^{\Delta-1} - (1-z)\mu + \sum_{j=\Delta}^\infty \frac{(1-y)^j}j
\end{align*}
Equivalently,
\begin{align*}
\sum_{j=\Delta}^\infty \frac{(1-y)^j}{j} &= \Delta\epsilon(1-y)^{\Delta-1} \\
\sum_{j=\Delta}^\infty \frac{(1-y)^{j-\Delta+1}}{j} &= \Delta\epsilon\\
\sum_{j=0}^\infty \frac{(1-y)^{j+1}}{j+\Delta} & = \Delta\epsilon\\
(1-y)\cdot \Phi(1-y, 1, \Delta) & = \Delta\epsilon.
\end{align*}
Note that the function $g(y)=(1-y)\cdot \Phi(1-y, 1, \Delta)$ is continuous and strictly decreasing for $y\in(0,1)$, moreover, $\lim_{y \rightarrow 0^+} g(y) = \infty$ and $g(1) = 0$. Therefore, there exists a unique value of $y\in(0,1)$ that satisfies the equation and, by definition, such value is $\hat{w}_1$. While an exact computation of $\hat{w}_1$ would be challenging, we are only interested  in small enough values of $\epsilon$ with respect to $\Delta$. We then compute $\hat{w}_1$ within an error term. Let $y_1=1-\Delta^2\epsilon + \frac{\Delta^5}{\Delta+1}\epsilon^2$ and $y_2 = y_1 - \frac{\Delta^7(\Delta^2+2\Delta-1)}{(\Delta+1)^2(\Delta+2)} \epsilon^3$. We have $y_1,y_2\in(0,1)$ for a sufficiently small $\epsilon$. In particular, we will select $\epsilon < \frac{1}{2\cdot 10^3 \cdot (3\Delta)^{7\Delta}}$. Note that $y_1 \geq y_2 \geq \nicefrac12$. Moreover,
\begin{align*}
    \sum_{j=3}^\infty \frac{(1-y_1)^{j+1}}{j+\Delta} & \leq \sum_{j=3}^\infty (1-y_1)^{j+1} = (1-y_1) \left(\sum_{j=0}^\infty (1-y_1)^j - \sum_{j=0}^2 (1-y_1)^j\right)\\
    & = (1-y_1) \left(\frac{1}{y_1} - \frac{1-(1-y_1)^3}{y_1}\right) = \frac{(1-y_1)^4}{y_1}\\
    & \leq 2(1-y_1)^4 = 2\left(\Delta^2\epsilon - \frac{\Delta^5}{\Delta+1} \epsilon^2\right)^4 \le 2\left(\Delta^2\epsilon\right)^4 \le 2\Delta^8  \epsilon^4,
\end{align*}
since $\Delta^2 \epsilon > \Delta^{4} \epsilon^2 \ge \frac{\Delta^5}{\Delta+1} \epsilon^2$ for the chosen $\epsilon$. By using that $\frac{\Delta^6(\Delta^2 + 2\Delta-1)}{(\Delta+1)^2(\Delta+2)} \geq \frac{\Delta^5}{12}$ for $\Delta\geq 1$, we have,
\begin{align*}
g(y_1) &= \frac{\Delta^2\epsilon - \frac{\Delta^5}{\Delta+1}\epsilon^2}{\Delta} + \frac{\left(\Delta^2\epsilon - \frac{\Delta^5}{\Delta+1}\epsilon^2\right)^2}{\Delta+1} + \frac{\left(\Delta^2\epsilon - \frac{\Delta^5}{\Delta+1}\epsilon^2\right)^3}{\Delta+2} + \sum_{j=3}^\infty \frac{(1-y_1)^{j+1}}{j+\Delta}\\ 
&\leq \Delta\epsilon - \frac{\Delta^4}{\Delta+1}\epsilon^2 + \frac{\left(\Delta^2\epsilon - \frac{\Delta^5}{\Delta+1}\epsilon^2\right)^2}{\Delta+1} + \frac{\Delta^6\epsilon^3}{\Delta+2} + 2\Delta^8 \epsilon^4 \\
&\leq \Delta\epsilon - \frac{\Delta^4}{\Delta+1}\epsilon^2 + \frac{\Delta^4}{\Delta+1}\epsilon^2 - \frac{2\Delta^7}{(\Delta+1)^2}\epsilon^3 + \frac{\Delta^{10}}{(\Delta+1)^3} \epsilon^4 + \frac{\Delta^6}{\Delta+2}\epsilon^3
+ 2\Delta^8 \epsilon^4 \\ 
& \leq\Delta\epsilon - \frac{\Delta^6(\Delta^2+2\Delta-1)}{(\Delta+1)^2(\Delta+2)} \cdot \epsilon^3 + 3\Delta^8 \epsilon^4\\
&\leq\Delta\epsilon - \frac{\Delta^5}{12} \cdot \epsilon^3 + 3\Delta^8 \epsilon^4 \\
& < \Delta \epsilon,
\end{align*}
where the last inequality holds because $\epsilon < \frac1{36 \cdot \Delta^3}$.
Now, observe that $\frac{2\Delta^{12}(\Delta^2 + 2\Delta-1)}{(\Delta+1)^3(\Delta+2)} \leq 6\Delta^{10}$ and $\frac{\Delta^8\left( \Delta^4 + 5\Delta^3 + 4\Delta^2-8\Delta+2\right)}{(\Delta+1)^3(\Delta+2)(\Delta+3)} \geq \frac{\Delta^7}{96}$ for $\Delta\geq 1$. Note also that $(1-x)^4 \geq (1-8x)$ for $x\in (0,1/8)$. We lower bound $g(y_2)$ by truncating the sum to the fourth term, 
\begin{align*}
g(y_2) & > \frac{\Delta^2\epsilon - \frac{\Delta^5}{\Delta+1}\epsilon^2 + \frac{\Delta^7(\Delta^2+2\Delta-1)}{(\Delta+1)^2(\Delta+2)}\epsilon^3}{\Delta} + \frac{\left(\Delta^2\epsilon - \frac{\Delta^5}{\Delta+1}\epsilon^2 + \frac{\Delta^7(\Delta^2+2\Delta-1)}{(\Delta+1)^2(\Delta+2)}\epsilon^3\right)^2}{\Delta+1} \\ 
& \quad + \frac{\left(\Delta^2\epsilon - \frac{\Delta^5}{\Delta+1}\epsilon^2 + \frac{\Delta^7(\Delta^2+2\Delta-1)}{(\Delta+1)^2(\Delta+2)}\epsilon^3\right)^3}{\Delta+2} +  \frac{\left(\Delta^2\epsilon - \frac{\Delta^5}{\Delta+1}\epsilon^2 + \frac{\Delta^7(\Delta^2+2\Delta-1)}{(\Delta+1)^2(\Delta+2)}\epsilon^3\right)^4}{\Delta+3}\\
& \geq \Delta\epsilon - \frac{\Delta^4}{\Delta+1}\epsilon^2 + \frac{\Delta^6(\Delta^2+2\Delta-1)}{(\Delta+1)^2(\Delta+2)}\epsilon^3  + \frac{\left(\Delta^2\epsilon - \frac{\Delta^5}{\Delta+1}\epsilon^2 + \frac{\Delta^7(\Delta^2+2\Delta-1)}{(\Delta+1)^2(\Delta+2)}\epsilon^3\right)^2}{\Delta+1} \\
&\quad + \frac{\left(\Delta^2\epsilon - \frac{\Delta^5}{\Delta+1}\epsilon^2 \right)^3}{\Delta+2} + \frac{\left(\Delta^2\epsilon - \Delta^4\epsilon^2 \right)^4}{\Delta+3}\\
& \geq \Delta\epsilon - \frac{\Delta^4}{\Delta+1}\epsilon^2 + \frac{\Delta^6(\Delta^2+2\Delta-1)}{(\Delta+1)^2(\Delta+2)}\epsilon^3 +\frac{\Delta^4}{\Delta+1}\epsilon^2 + \frac{\Delta^{10}}{(\Delta+1)^3}\epsilon^4 - \frac{2\Delta^7}{(\Delta+1)^2}\epsilon^3 \\
& \quad + \frac{2\Delta^9(\Delta^2+2\Delta-1)}{(\Delta+1)^3(\Delta+2)}\epsilon^4 - \frac{2\Delta^{12}(\Delta^2+2\Delta-1)}{(\Delta+1)^3(\Delta+2)} \epsilon^5 \\
& \quad + \frac{\Delta^6}{\Delta+2}\epsilon^3 -\frac{3\Delta^9}{(\Delta+1)(\Delta+2)}\epsilon^4 - \frac{\Delta^{15}}{(\Delta+1)^3(\Delta+2)} \epsilon^6 + \frac{(\Delta^2\epsilon)^4 (1-\Delta^2\epsilon)^4}{\Delta+3} \\
& \geq \Delta\epsilon + \frac{\Delta^{10}}{(\Delta+1)^3}\epsilon^4 + \frac{2\Delta^9(\Delta^2+2\Delta-1)}{(\Delta+1)^3(\Delta+2)}\epsilon^4 -\frac{3\Delta^9}{(\Delta+1)(\Delta+2)}\epsilon^4 + \frac{\Delta^8\epsilon^4(1-8\Delta^2\epsilon)}{\Delta+3} - 7\Delta^{10}\epsilon^5\\
& \geq \Delta\epsilon + \frac{\Delta^8\left( \Delta^4 + 5\Delta^3 + 4\Delta^2-8\Delta+2\right)}{(\Delta+1)^3(\Delta+2)(\Delta+3)} \cdot \epsilon^4 - 15\Delta^{10}\epsilon^5 \\
& \geq \Delta\epsilon +\frac{\Delta^7}{96} \epsilon^4 - 15\Delta^{10}\epsilon^5\\
&> \Delta\epsilon,
\end{align*}
where the last inequality follows since $\epsilon< \frac{1}{2\cdot 10^3 \cdot \Delta^3}$. Since $g(y)$ is continuous and decreasing, it must be $y_2 < \hat{w}_1 < y_1$, and therefore, %
\[
\hat{w}_1 = 1 - \Delta^2\epsilon + \frac{\Delta^5}{\Delta+1}\epsilon^2 - R(\Delta, \epsilon),
\]
where $R(\Delta,\epsilon)=y_1 - \hat{w}_1 \leq y_1 - y_2 \leq 3\Delta^6 \epsilon^3$. Moreover, by inverting the relation $\hat{w}_1=e^{-(1-w_2)\mu}$ we get, $w_2=1+\frac{\ln\hat{w}_1}{\mu}$. Observe that $w_2$ is well-defined because $y_2 \geq 1/2 > e^{-\mu}$, and thus $\frac{\ln(\hat{w}_1)}{\mu}>-1$.

Let $\alpha=f(1)+\hat{f}(1)-\hat{f}(w_2)-f'(1-\hat{w}_1)\hat{w}_1-f(1-\hat{w}_1)$. By \Cref{thm:bg} the maximum matching in our graph has size at most $(\alpha+o(1))n$. Note that $\hat{f}(w_2)=\exp\left(\frac{\ln\hat{w}_1}{\mu}\mu\right)=\hat{w}_1$, and $f(1)=\hat{f}(1)=1$. We have,
\begin{align*}
\alpha & = 2 - \hat{f}(w_2) - f'(1-\hat{w}_1)\hat{w}_1 - f(1-\hat{w}_1) \\
& = 2 - \hat{w}_1 - \Delta\epsilon\hat{w}_1(1-\hat{w}_1)^{\Delta-1}+\hat{w}_1\ln(\hat{w}_1) + \hat{w}_1(1-\hat{w}_1)^{\Delta}\cdot \Phi(1-\hat{w}_1,1,\Delta)+\\
&\quad \epsilon-\frac{1}{\Delta} - \hat{w}_1 \left( \ln \hat{w}_1 + (1-\hat{w}_1)^{\Delta+1}\cdot\Phi(1-\hat{w}_1,1,\Delta+1) \right) - (1-\hat{w}_1)+\frac{(1-\hat{w}_1)^{\Delta+1}}{\Delta}-\epsilon(1-\hat{w}_1)^\Delta\\
&=1-\Delta\epsilon\hat{w}_1(1-\hat{w}_1)^{\Delta-1} + \hat{w}_1(1-\hat{w}_1)^{\Delta}\cdot \Phi(1-\hat{w}_1,1,\Delta)+\\
&\quad\epsilon-\frac{1}{\Delta}-\hat{w}_1 (1-\hat{w}_1)^{\Delta+1}\cdot\Phi(1-\hat{w}_1,1,\Delta+1)+\frac{(1-\hat{w}_1)^{\Delta+1}}{\Delta}-\epsilon(1-\hat{w}_1)^\Delta.
\end{align*}
For each $x \in [0,1]$, and for each integer $k \ge 2$, we have $x^{k+1} \cdot \Phi(x,1,k+1) +  \frac{x^k}k= x^k \cdot \Phi(x,1,k)$. Thus, $(1-\hat{w}_1)^{\Delta}\Phi(1-\hat{w}_1,1,\Delta) = (1-\hat{w}_1)^{\Delta+1}\Phi(1-\hat{w}_1,1,\Delta+1) + \frac{(1-\hat{w}_1)^{\Delta}}{\Delta}$. Therefore, by simplifying,
\begin{align*}
\alpha & = 1 + \epsilon - \frac{1}{\Delta} -\Delta\epsilon\hat{w}_1(1-\hat{w}_1)^{\Delta-1} + \frac{\hat{w}_1(1-\hat{w}_1)^\Delta+(1-\hat{w}_1)^{\Delta+1}}{\Delta} - \epsilon(1-\hat{w}_1)^\Delta \\
& = 1 + \epsilon - \frac{1}{\Delta} -\Delta\epsilon\hat{w}_1(1-\hat{w}_1)^{\Delta-1} + \frac{(1-\hat{w}_1)^{\Delta}}{\Delta} - \epsilon(1-\hat{w}_1)^\Delta \\
&=1+\epsilon-\frac{1}{\Delta}+(1-\hat{w}_1)^{\Delta-1}\left(-\Delta\epsilon\hat{w}_1 + \frac{1-\hat{w}_1}{\Delta} - \epsilon(1-\hat{w}_1)\right)\\
&=1 +\epsilon-\frac{1}{\Delta}-(1-\hat{w}_1)^{\Delta-1}\left(\epsilon\cdot(1+(\Delta-1)\hat{w}_1) - \frac{1-\hat{w}_1}{\Delta}\right).
\end{align*}
Note that $(1-\hat{w}_1)^{\Delta-1}>0$ and moreover, we have, by using $\hat{w}_1=1-\Delta^2\epsilon+\frac{\Delta^5}{\Delta+1}\epsilon^2 - R(\Delta,\epsilon)$, 
\begin{align*}
&\epsilon(1+(\Delta-1)\hat{w}_1) - \frac{1-\hat{w_1}}{\Delta} \\
&= \epsilon\left(1 + (\Delta-1)\left(1 - \Delta^2\epsilon + \frac{\Delta^5}{\Delta+1}\epsilon^2 - R(\Delta,\epsilon)\right)\right) - \frac{\Delta^2\epsilon - \frac{\Delta^5}{\Delta+1}\epsilon^2 + R(\Delta,\epsilon)}{\Delta}\\
& = \epsilon\left(\Delta - (\Delta-1)\Delta^2\epsilon +\frac{\Delta^5(\Delta-1)}{\Delta+1}\epsilon^2 - (\Delta-1)R(\Delta,\epsilon) \right) - \Delta\epsilon + \frac{\Delta^4}{\Delta+1}\epsilon^2- \frac{R(\Delta,\epsilon)}{\Delta}\\
& =\frac{\Delta^2}{\Delta+1} \cdot \epsilon^2 - R_2(\Delta, \epsilon), 
\end{align*}
with $R_2(\Delta,\epsilon) = \frac{R(\Delta,\epsilon)}{\Delta} + \epsilon(\Delta-1)R(\Delta,\epsilon)  - \frac{\Delta^5(\Delta-1)}{(\Delta+1)}\epsilon^3$ and therefore, $|R_2(\Delta,\epsilon)|\leq 7\Delta^5\epsilon^3$.
Thus, since $\epsilon < \frac{1}{7\Delta^3(\Delta+1)}$, we have $\alpha < 1 -\frac{1}{\Delta} + \epsilon$. In particular,
\begin{align*}
\alpha &=  1 - \frac1{\Delta} + \epsilon - \left(\Delta^2 \cdot \epsilon - \frac{\Delta^5}{\Delta+1} \cdot \epsilon^2 + R(\Delta,\epsilon)\right)^{\Delta - 1} \cdot \left(\frac{\Delta^2}{\Delta+1} \cdot \epsilon^2 - R_2(\Delta,\epsilon)\right) \\
&=1 - \frac1{\Delta} + \epsilon - \left(\Delta^{2\Delta-2} \cdot \epsilon^{\Delta-1} +R_3(\Delta,\epsilon)\right) \cdot \left(\frac{\Delta^2}{\Delta+1} \cdot \epsilon^2 - R_2(\Delta,\epsilon)\right) \\
&= 1 - \frac1{\Delta} + \epsilon - \frac{\Delta^{2\Delta}}{\Delta+1} \cdot \epsilon^{\Delta+1} + R_4(\Delta,\epsilon),
\end{align*}
where $R_3(\Delta,\epsilon)$ is the sum of $3^{\Delta-1}-1$ terms, each no larger than $(3\Delta)^{6(\Delta-1)}\epsilon^{\Delta}$ in absolute value, and thus, $|R_3(\Delta,\epsilon)|\leq (3\Delta)^{7(\Delta-1)}\epsilon^\Delta$. Moreover,
\begin{align*}
|R_4(\Delta,\epsilon)| &\leq \left|\frac{\Delta^2 \epsilon^2 R_3(\Delta,\epsilon)}{\Delta+1}\right| + \left|R_3(\Delta,\epsilon)R_2(\Delta,\epsilon)\right| + \left|R_2(\Delta,\epsilon)\Delta^{2\Delta-2}\epsilon^{\Delta-1}\right|\\
& \leq (3\Delta)^{7\Delta-5}\epsilon^{\Delta+2} + 7(3\Delta)^{7\Delta-2}\epsilon^{\Delta+3} + 7\Delta^{2\Delta+3}\epsilon^{\Delta+2}\\
&\leq 15\cdot (3\Delta)^{7\Delta} \cdot \epsilon^{\Delta+2}.
\end{align*}
Thus, $\psi_\Delta(\epsilon) = \frac{\Delta^{2\Delta}}{\Delta+1} \epsilon^{\Delta+1} - R_4(\Delta,\epsilon)> 0$ since $\epsilon < \frac{1}{15\cdot (3\Delta)^{7\Delta}}$.
\end{appendixproof}
We can now use the previous Lemma to prove the main result of this Section. %
\begin{theoremrep}\label{thm:extremality-arbitrary-distr}
Suppose that there exists a constant $i \ge 1$ such that $\Pr[P\leq j] = 1- \frac1j$ for each $j \in \{1,2,\ldots,i-1\}$, and $\Pr[P\leq i] > 1 - \frac1{i}$. Then, with probability $1-o(1)$,  the instance $\irregularcuckoohashing(n,n,P)$ admits no matching that matches at least $\Pr[P\leq i] \cdot n - o(n)$ users of degree $\le i$.
\end{theoremrep}
\begin{appendixproof}
Consider first the case $i=1$. Note that, if $\Pr[P=0] > 0$, then, from a simple concentration bound, we have that, with probability $1-o(1)$, at least $\Pr[P=0] \cdot n - o(n)$ users of degree $\leq i$ cannot be matched. Consider now the case $\Pr[P=0]=0$ and $\Pr[P=1] = \epsilon > 0$. With probability $1-o(1)$, there are at least $(\epsilon - o(1))n$ users of degree 1. Let $S$ be the set of these users. Consider the generic ad $a$. We compute the probabilities that this ad has $0$ or $1$ neighbors in $S$: $\Pr[a \text{ has no neighbor in } S] = \left(1-\frac1n\right)^{|S|} = e^{-\epsilon} \pm o(1)$, and $\Pr[a \text{ has exactly 1 neighbor in } S] = |S| \cdot \frac1n \cdot \left(1-\frac1n\right)^{|S|-1} = \epsilon \cdot e^{-\epsilon} \pm o(1)$. Thus, if we let $T$ be the set of ads that have $2$ or more neighbors in $S$, we have that, with probability $1-o(1)$, $|T| \ge (1 - (\epsilon + 1)\cdot e^{-\epsilon}) \cdot n - o(n)$. Given that each user in $S$ has degree $1$, in each matching $M$, each ad in $T$ contributes to the count of users unmatched by $M$ with (at least) a unique user. Therefore, at least $(1 - (\epsilon + 1)\cdot e^{-\epsilon})\cdot n - o(n)$ users of degree $1$ will remain unmatched.

Consider now the case $i\geq 2$. Let $\rho=\Pr[P=i] - \frac{1}{i(i-1)} > 0$, and let $0<\epsilon < \rho$ be a small enough constant with respect to $i$ such that \Cref{lem:eps-mass-extremality} holds. Let $I$ be an instance sampled according to distribution $P$. While sampling $I$, we build the instance $I'$ in the following way: (i) if the sampled user has degree $>i$, we change its degree to $0$, and (ii) if the sampled user has degree exactly $i$, with probability $\frac{\rho - \epsilon}{\Pr[P=i]}$ we change its degree to $0$. Let $A$ be the set of users that had their degree changed from $i$ to $0$. Let $M(I)$ (resp. $M(I')$) denote the number of users of degree $\leq i$ that can be matched in instance $I$ (resp. $I'$). We have $M(I) \leq M(I') + |A|$. Note that instance $I'$ is distributed as the instance of \Cref{def:eps-mass-instance} with $\Delta=i$. Therefore by \Cref{lem:eps-mass-extremality}, with probability $1-o(1)$, $M(I')\leq (1-\frac{1}{i}+\epsilon-\psi_i(\epsilon) +o(1))n$ for some constant $\psi_i(\epsilon)>0$. Moreover, by a simple concentration bound (e.g., \Cref{fact:chernoff-hoeffding}), with probability $1-o(1)$, $|A|\leq (\rho - \epsilon + o(1))n$. Therefore, with probability $1-o(1)$, $M(I)\leq (1-\frac{1}{i}+\epsilon-\psi_i(\epsilon)+\rho-\epsilon+o(1))n=(\Pr[P \leq i] - \psi_i(\epsilon) + o(1))n$. %
\end{appendixproof}

Our distribution, then, is extremal among the irregular cuckoo hashing distributions that make it possible to match $n-o(n)$ users to the $n$ ads; indeed, it has $\Pr[P\le i] = 1 - \frac1i$ for each integer $i \ge 1$, and thus it greedily maximizes, for each $d = 0, 1, 2, \ldots$, the number of users of degree $d$ under the constraint that a $1-o(1)$ fraction of the users of degree at most $d$ can be matched. 

\section{Generalization to Non-Equibipartite Graphs}%
\label{sec:uneven}
In this Section, we consider a more general case, where we have $n$ ads and $u \cdot n$ users, for a constant $u \in (0,1]$. As in the instance of Definition~\ref{def:instance}, we will greedily add users of smallest degree, while guaranteeing the existence of a quasi-complete matching. Our generalized instances have varying competitive ratios; we will show that the worst competitive ratio is obtained at $u = 1$,\footnote{We do not consider explicitly the case where there are more users than ads, that is, the case $u > 1$. This case can be easily reduced to the case of $u = 1$, i.e., to the case of Definition~\ref{def:instance}. %
Indeed, if $u > 1$ and one aims to greedily add (fractions of) users of the minimum degree that allow a quasi-complete matching to exist, the first $n$ users will have the degrees given by the distribution of Definition~\ref{def:instance}. As we proved in Theorem~\ref{thm:mm}, these first $n$ users induce a matching that is quasi-perfect, and thus quasi-complete from the ads side. The next minimum user degree which guarantees a quasi-complete matching is then $0$. Thus, if $u > 1$, the greedy approach would produce $(u-1)\cdot n$ users of degree $0$, and   $\frac1{d\cdot(d-1)} \cdot n$ users of degree $d$ for each $d \ge 2$. The resulting graph will then induce online and  offline matching sizes  equal to those of Definition~\ref{def:instance}.} that is, with the original instance of Definition~\ref{def:instance} (see Figure~\ref{fig:DuCR} and Theorem~\ref{thm:optimal_u}).
\begin{definition}[Generalized Instance]\label{def:generalized-u-instance}
Let $u\in(0,1]$ be a constant. Consider the irregular cuckoo hashing instance $\irregularcuckoohashing(\lceil u \cdot n \rceil, n, D_u)$. The distribution $D_u$ is defined as,
\begin{itemize}
\item if $u=1$, $D_u$ is equal to the instance of \Cref{def:instance}, that is, $\Pr[D_u=d] = \frac{1}{d \cdot (d-1)}$ for $d\geq 2$,
\item if $u\in(0,1)$, let $\Delta$ be the minimum integer such that $u \le 1 - \frac1{\Delta}$, that is, let $\Delta = \left\lceil\frac1{1-u}\right\rceil$. Then, $\Delta \ge 2$. For each $d \in \{2,3,\ldots,\Delta-1\}$, we set $\Pr[D_u=d] = \frac{1}{u\cdot d \cdot (d-1)}$, and $\Pr[D_u=\Delta] = 1 - \frac{1}{u} \cdot \left(1-  \frac{1}{\Delta-1}\right)$.
\end{itemize}
\end{definition}
In particular, the distribution $D_u$ is obtained by conditioning the $D$ of \Cref{def:instance} on an event of probability $u$, chosen to include the smallest possible values of $D$'s support. Observe that, for each $u \in \left(0,\frac12\right]$, the distribution $D_u$ assigns probability $1$ to degree $2$. Observe, further, that as $u$ increases from $\frac{d-1}d$ to $\frac d{d+1}$, for any integer $d \ge 2$, $D_u$ assigns more and more probability to degree $d+1$.

\smallskip

The following claim is a corollary of the analysis of the maximum matching for the case $u=1$.
\begin{corollaryrep}\label{cor:quasi-complete-u}
The maximum matching of the instance of \Cref{def:generalized-u-instance} has size $u \cdot n - o(n)$ with probability $1-o(1)$; thus, the resulting graph admits a quasi-complete matching.
\end{corollaryrep}
\begin{appendixproof}
We only need to prove this for $u<1$. Let us add $n- \lceil u \cdot n\rceil$ new users: $\frac{n}{\Delta}$ users of degree 0 and $(1-u-\frac{1}{\Delta})n$ users of degree $\Delta$ (rounded to be integers). With an argument similar to \Cref{obs:graph-concentrated}, for each $d\in\{2,3,\dots, \Delta\}$ the number of users of degree $d$ is concentrated around the expected value. Conditioning on this event that happens with probability $1-o(1)$, the number of edges is $\sum_{d=2}^\Delta \frac{n}{d-1}\pm o(n^{\nicefrac45})=(H_{\Delta-1} \pm o(n^{\nicefrac{-1}5}))n$. Therefore, with probability $1-o(1)$, the degrees of the ads follow a Poisson distribution of parameter $H_{\Delta-1}$ as in \Cref{obs:graph-concentrated}. Conditioning on this event, we obtain the same configuration model as the one of \Cref{lem:karp-sipser-conf-model}, therefore the resulting graph has a matching of size at least $(1-\frac{1}{\Delta}-o(1))n$, and thus, the original $D_u$ instance has a matching of size at least $(1-\frac{1}{\Delta}- 1+u+\frac{1}{\Delta}-o(1))n=(u-o(1))n$, with probability $1-o(1)$.
\end{appendixproof}

One can prove that the $D_u$ construction is extremal using the same argument we developed in the last section for the case $u = 1$ (i.e., moving any constant mass from higher degrees to smaller degrees makes the quasi-complete matching lose a constant fraction of the edges). 

\begin{corollaryrep}
Suppose that there exists a constant $i \ge 1$ such that $\Pr[P\leq j] = \Pr[D_u \le j]$ for each $j \in \{1,2,\ldots,i-1\}$, and $\Pr[P\leq i] > \Pr[D_u \le i]$. Then in the instance $\irregularcuckoohashing(\lceil u\cdot n\rceil,n,P)$, with probability $1-o(1)$, there exists no matching that matches at least $\Pr[P \leq i] \cdot u \cdot n - o(n)$ users of degree $\le i$.
\end{corollaryrep}
\begin{appendixproof}
We only need to consider the case $u<1$. Consider the case $i=1$ first. Clearly, if $\Pr[P=0]>0$, a constant fraction of vertices of degree $\leq i$ cannot be matched. If, instead, $\Pr[P=0]=\Pr[D_u=0]=0$, and $\Pr[P=1] = \epsilon > 0$ then, with an argument similar to the one used in \Cref{thm:extremality-arbitrary-distr}, there are, with probability $1-o(1)$, $(\epsilon+o(1))\cdot u\cdot n$ users of degree $1$, and at least $(1-(u\cdot \epsilon+1)e^{-u\cdot \epsilon})n - o(n)$ of these users will not be matched. Thus, the claim is proved for $i=1$.

Consider now $i>1$. Since $\Pr[D_u \leq \Delta]=1$, it must be $1<i<\Delta$. Let $\rho = \Pr[P=i] - \Pr[D_u =i] > 0$ and let $0<\epsilon<u\cdot \rho$ be a small enough constant with respect to $i$. With probability $1-o(1)$ the number of users of degree $\{2,3,\dots, \Delta\}$ is concentrated around the expectation. Condition on this event, add $n-\lceil u\cdot n\rceil$ users of degree 0 and remove all edges incident to users of degree larger than $i$ and to $(\rho u-\epsilon \pm o(n^{\nicefrac{-1}5}))n$ users of degree $i$. The resulting graph has $(\frac{1}{d(d-1)}\pm o(n^{\nicefrac{-1}{5}}))n$ users of degree $d\in\{2,\dots, i-1\}$, $(\frac{1}{u\cdot i(i-1)}+\rho \pm o(n^{\nicefrac{-1}5}))u\cdot n - (\rho u - \epsilon \pm o(n^{\nicefrac{-1}5}))n=(\frac{1}{i(i-1)}+\epsilon\pm o(n^{\nicefrac{-1}5}))n$ users of degree $i$, and the number of users of degree 0 is,
\begin{align*}
&n-\lceil un \rceil + (\rho u - \epsilon \pm o(n^{\nicefrac{-1}5}))n + \left(1-\sum_{j=2}^{i}\frac{1}{ui(i-1)}-\rho \pm  o(n^{\nicefrac{-1}5})\right)u\cdot n\\
&= n-\lceil un \rceil + (\rho u - \epsilon \pm o(n^{\nicefrac{-1}5}))n + \left(1-\frac{1}{u}\left(1 - \frac{1}{i}\right)-\rho \pm o(n^{\nicefrac{-1}5})\right)u\cdot n\\
& = n\left(\frac{1}{i} - \epsilon \pm o(n^{\nicefrac{-1}5})\right).
\end{align*}
Thus, the resulting graph is the configuration model described in \Cref{lem:eps-mass-extremality} with $\Delta=i$ and therefore, for a small enough constant $\epsilon$ with respect to $i$, with probability $1-o(1)$ the maximum matching has size at most $(1-\frac{1}{i}+\epsilon-\psi_i(\epsilon) + o(1))n$, for $\psi_i(\epsilon)>0$. Thus, in the instance given by $P$, with probability $1-o(1)$, the number of users of degree $\leq i$ that can be matched is at most $(1-\frac{1}{i}+\epsilon-\psi_i(\epsilon) + o(1))n + (\rho u - \epsilon + o(1))n = (\Pr[P \leq i] - \frac{\psi_i(\epsilon)}{u})u\cdot n + o(n)$. Note that $0<\psi_i(\epsilon)\leq \epsilon < u\rho$ and $\nicefrac{\psi_i(\epsilon)}{u} > 0$, which concludes the proof.
\end{appendixproof}
We now move on to studying the performance of online algorithms on the generalized instance of \Cref{def:generalized-u-instance}. We start with a technical lemma that relates the expected number of users to capture a new ad with distribution $D_u$, and with distribution $D$. This Lemma will be used to %
bound the number of users needed to capture a new ad with distribution $D_u$, in terms of the number of users needed with $D$ --- we need this Lemma because the integral required to determine the number of users to capture an ad with $D_u$ is not easy to solve, while the one for $D$ has already been solved in the proof of \Cref{thm:om}.
\begin{lemmarep}\label{lem:partial_ub}
Fix $u \in \left(\frac12,1\right)$ and let $\Delta = \left\lceil\frac1{1-u}\right\rceil$, $c_i = \frac1{u \cdot i\cdot(i-1)}$ for $i \in \{2,\ldots,\Delta-1\}$ and $c_\Delta = 1-\sum_{i=2}^{\Delta-1} \frac1{u \cdot i \cdot (i-1)} = 1 - \frac{\Delta-2}{u \cdot(\Delta-1)}$.\footnote{Observe that $c_i = \Pr[D_u = i] \; \forall i \in \{2,\ldots,\Delta\}$, where $D_u$ is the distribution of Definition~\ref{def:generalized-u-instance}.}
    Moreover, let $x \in \left[0,u\right)$. Then, 
    \[
    \frac1{1- \sum_{i=2}^\Delta \left(c_i \cdot x^i\right)} \le  \frac{1-\frac{1}{8(1+\ln 2)\Delta} \cdot \left[ x \ge \frac{u}{2}\right]}{\left(1-\frac xu\right)\left(1-\ln\left(1-\frac xu\right)\right)}.
    \]
\end{lemmarep}
\begin{appendixproof}
Observe that $u > \frac12$ implies $\Delta \ge 3$. Then,
\begin{align*}
    \frac1{1- \sum_{i=2}^\Delta \left(c_i \cdot x^i\right)} &=
     \frac1{1- \frac1u \cdot \sum_{i=2}^{\Delta-1} \frac{x^i}{i\cdot(i-1)} -  c_\Delta \cdot x^{\Delta} } \\
    &=\frac1{1-\frac 1u\cdot\left(\frac{x^2}2+\sum_{i=3}^{\Delta-1} \frac{x^i}{i \cdot (i-1)}+u \cdot c_\Delta \cdot x^{\Delta}\right)}\\
    &=\frac1{1-\frac 1u\cdot\left(\frac{x^2}2+\sum_{i=3}^{\Delta-1} \frac{x^i}{i \cdot (i-1)}+\left(u - \frac{\Delta-2}{\Delta-1}\right) \cdot x^{\Delta}\right)}.
\end{align*}
Observe that $\Delta \ge \frac1{1-u}$ so that $1-u \ge \frac1{\Delta}$ and $u \le  \frac{\Delta-1}{\Delta}$. Therefore, $u - \frac{\Delta-2}{\Delta-1} \le \frac{\Delta-1}{\Delta} - \frac{\Delta-2}{\Delta-1} =\frac1{\Delta\cdot(\Delta-1)}$. Then,
\begin{align*}
    \frac1{1- \sum_{i=2}^\Delta \left(c_i \cdot x^i\right)} &\le \frac1{1-\frac 1u\cdot\left(\frac{x^2}2+\sum_{i=3}^{\Delta-1} \frac{x^i}{i \cdot (i-1)}+\frac1{\Delta\cdot(\Delta -1)}\cdot x^{\Delta}\right)}\\
    &=\frac1{1-\frac 1u\cdot\left(\frac{x^2}2+\sum_{i=3}^{\Delta} \frac{x^i}{i \cdot (i-1)}\right)}\\
    &=\frac1{1+\frac{(x/u)^2 -x^2/u}2-\frac{(x/u)^2}2-\frac1u\sum_{i=3}^{\Delta} \frac{x^i}{i \cdot (i-1)}}\\
    &\le\frac1{1+\frac{(x/u)^2 -x^2/u}2-\frac{(x/u)^2}2-\sum_{i=3}^{\Delta} \frac{(x/u)^i}{i \cdot (i-1)}}\\
     &=\frac1{1+\frac{(x/u)^2 -x^2/u}2-\sum_{i=2}^{\Delta} \frac{(x/u)^i}{i \cdot (i-1)}}\\
     &=\frac1{1+\frac{u^{-2} - u^{-1}}2 \cdot x^2-\sum_{i=2}^{\Delta} \frac{(x/u)^i}{i \cdot (i-1)}}\\
     &\le\frac1{1-\sum_{i=2}^{\infty} \frac{(x/u)^i}{i \cdot (i-1)}+\frac{u^{-2} - u^{-1}}2 \cdot x^2}\\
     &=\frac1{\left(1-\frac xu\right)\left(1-\ln\left(1-\frac xu\right)\right)+\frac{u^{-2} - u^{-1}}2 \cdot x^2}
\end{align*}
where the last step follows from $1 - \sum_{i=2}^{\infty} \frac{y^i}{i\cdot(i-1)} = (1-y) (1-\ln(1-y))$ for $y\in(0,1)$ (see Equation~\ref{eqn:qi}). %
We have $\frac{u^{-2} - u^{-1}}{2} \cdot x^2 \ge 0$ for each $x \ge 0$, and $\frac{u^{-2} - u^{-1}}{2} \cdot x^2 \ge \frac{1-u}{8}$ for each $x \ge \frac{u}{2}$, thus, $\frac{u^{-2} - u^{-1}}{2} \cdot x^2 \ge \frac{1-u}{8} \cdot \left[x \ge \frac{u}{2}\right]$.  Then,
\begin{align*}
    \frac1{1- \sum_{i=2}^\Delta \left(c_i \cdot x^i\right)}
    &\le\frac1{\left(1-\frac xu\right)\left(1-\ln\left(1-\frac xu\right)\right)+\frac{1-u}{8} \cdot \left[x \ge \frac{u}{2}\right]}.
\end{align*}
Moreover, given that $\left(1- \frac xu\right)\left(1-\ln\left(1-\frac xu\right)\right)$ is decreasing in $\left[0,u\right)$, we have that $\left(1-\frac xu\right)\left(1-\ln\left(1-\frac xu\right)\right) \le \frac{1 + \ln 2}2$ for $x \in \left[\frac u2, u\right)$; moreover, given that $u \le 1-\frac1{\Delta}$, we have that $\frac{1-u}8 \ge \frac1{8 \Delta}$. It follows that, for $x \in \left[\frac u2,u\right)$, $\frac{1-u}8 \ge  \left(1-\frac xu\right)\left(1-\ln\left(1-\frac  xu\right)\right) \cdot \frac2{1+\ln 2} \cdot \frac1{8\Delta}$. Consequently, for each $x \in [0,u)$,
\begin{align*}
    \frac1{1- \sum_{i=2}^\Delta \left(c_i \cdot x^i\right)}
    &\le \frac{1}{\left(1-\frac xu\right)\left(1-\ln\left(1-\frac xu\right)\right)\left(1+ \frac{1}{4(1+\ln 2)\Delta} \cdot \left[x \ge \frac{u}{2}\right]\right)}\\
    &\le\frac{1-\frac1{8(1+\ln 2)\Delta} \cdot \left[x \ge \frac{u}{2}\right]}{\left(1-\frac xu\right)\left(1-\ln\left(1-\frac xu\right)\right)},
\end{align*}
where the last inequality follows from $\frac1{1+a} \le 1-a/2$, for each $a \in [0,1]$.
\end{appendixproof}
We now study the performance of  online algorithms on the instance $D_u$ for $u<1$.
\begin{lemmarep}\label{lem:online-matching-u-instance}
Fix $u \in (0,1)$. Then, the online greedy algorithm, if run on the $D_u$ distribution of Definition~\ref{def:generalized-u-instance}, will match at least $\left(1-e^{1-e}+\frac{1-u}{500} \right)  \cdot u \cdot  n - O(\sqrt{n \ln n})$ ads in expectation. %
\end{lemmarep}
\begin{appendixproof}
We analyze the greedy algorithm that, whenever possible, matches the current user to any available ad. Let $\alpha=u(1-e^{1-e})$, $m=\lfloor \alpha \cdot n \rfloor$. Imagine to keep sampling users until we match $m$ of them, and let $S$ be the total number of users sampled to match $m$ of them. Define, for $i\in\{0,1,\dots,m-1\}$, $x_i=\frac{i}{n}$. Let $c_i=\frac{1}{u\cdot i(i-1)}$ for $i\in\{2,\dots, \Delta-1\}$, and $c_{\Delta}=1 - \frac{\Delta-2}{u(\Delta-1)}$. Define $q_i=1-\sum_{j=2}^\Delta c_j\cdot x_i^j$ for $i\in\{0,\dots, m-1\}$.

We start by showing that $S$ is concentrated around the expected value. Of course, $S\geq m$, and therefore $\E[S] \geq m$. Note that $S=\sum_{i=0}^{m-1} S_i$, where $S_i\sim\geom(q_i)$ for $i\in\{0,\dots, m-1\}$, and $\min_{i\in\{0,\dots, m-1\}}(q_i) = q_{m-1}$. Let us lower bound $q_{m-1}$.
\begin{align*}
q_{m-1} &= 1 - \sum_{i=2}^{\Delta-1} \frac{1}{u\cdot i \cdot (i-1)}\cdot \left(\frac{m-1}{n}\right)^i - \left(1-\frac{\Delta-2}{u(\Delta-1)}\right) \left(\frac{m-1}{n}\right)^\Delta\\
& \geq  1 - \sum_{i=2}^{\Delta-1} \frac{1}{u\cdot i \cdot (i-1)}\cdot \alpha^i - \left(1-\frac{\Delta-2}{u(\Delta-1)}\right) \alpha^\Delta\\
&\geq 1 - \sum_{i=2}^{\Delta-1} \frac{1}{i \cdot (i-1)}\cdot (1-e^{1-e})^i -\alpha^\Delta  \geq 1 - \sum_{i=2}^{\infty} \frac{1}{i \cdot (i-1)}\cdot (1-e^{1-e})^i -\alpha^\Delta\\
&=e^{2-e} -\alpha^\Delta \geq e^{2-e} - (u(1-e^{1-e}))^{\frac{1}{1-u}}  > \frac{1}{4},
\end{align*}
where in the last equality we use that $1-\sum_{i\geq 2}\frac{x^i}{i(i-1)} = (1-x)(1-\ln(1-x))$ for $x\in(0,1)$ (see \Cref{eqn:qi}), and the last inequality follows from the fact that the maximum of the function $f(x)=(x(1-e^{1-e}))^{\frac{1}{1-x}}$ is not larger than $1/5$ in $(0,1)$. We can now apply \Cref{fact:geom-concentration}, where we choose $\lambda=1+\epsilon$ for $\epsilon=\sqrt{\frac{16\cdot \ln n}{m}}\in(0,1)$. It holds,
\begin{align*}
\Pr[S\geq (1+\epsilon)\E[S]] &\leq e^{-q_{m-1}\cdot \E[S] \cdot (\epsilon - \ln(1+\epsilon))} \leq e^{-\frac{m \cdot \epsilon^2}{16}}  \leq \frac{1}{n},
\end{align*}
where the second inequality holds because $\epsilon-\ln(1+\epsilon)\geq \frac{\epsilon^2}{4}$ for $\epsilon\in(0,1)$.

We now aim to upper bound the expected value of $S$. Since it is the sum of geometric random variables, we have,
\begin{align*}
\frac{\E[S]}{n} & = \frac{1}{n} \sum_{i=0}^{m-1} \frac{1}{q_i}= \frac{1}{n}\sum_{i=0}^{m-1}\frac{1}{1-\sum_{j=2}^\Delta c_j \cdot x_i^j}\\
& \leq \int_0^{m/n}\frac{1}{1-\sum_{j=2}^\Delta c_j \cdot x^j} \, dx  \leq \int_0^{\alpha}\frac{1}{1-\sum_{j=2}^\Delta c_j \cdot x^j} \, dx,
\end{align*}

where for the inequalities we use that $\frac{1}{n}\sum_{i=0}^{m-1}\frac{1}{q_i}$ is a left Riemann sum over $[0,\frac{m}{n}]$ and $\frac{1}{1-\sum_{j=2}^\Delta c_j x^j}$ is positive and increasing for $x\in[0,\frac{m}{n}]$. 

We now study the integral $\int_0^{\alpha}\frac{1}{1-\sum_{i=2}^\Delta c_i \cdot x^i} \, dx$. Consider the case $u\in(0,1/2]$ first, so that $\Delta=2$, we have,
\[
\int_0^{\alpha}\frac{1}{1-\sum_{i=2}^\Delta c_i \cdot x^i} \, dx = \int_0^{\alpha} \frac{1}{1-x^2} \, dx = [\arctanh(x)]_0^{\alpha} = \arctanh\left(u(1-e^{1-e})\right) \leq u\left(1-\frac{0.03}{\Delta}\right),
\]
where the last inequality follows from $\arctanh(x(1-e^{1-e})) \leq x(1 - 0.03)$ for $x\in[0,\frac12]$.\footnote{This can be proved by checking the condition at $x=0$ and $x=\frac12$ and by observing that $x(1 - 0.03)-\arctanh(x(1-e^{1-e}))$ is concave over $[0,\frac12]$.} Let us now consider the case $u\in(1/2, 1)$. By \Cref{lem:partial_ub}, we have, 
\begin{align*}
\int_0^{\alpha} \frac1{1- \sum_{i=2}^\Delta \left(c_i \cdot x^i\right)} \, dx 
& \le 
    \int_0^{\alpha} \frac{1-\frac{1}{8(1+\ln 2)\Delta} \cdot \left[ x \ge \frac{u}{2}\right]}{\left(1-\frac xu\right)\left(1-\ln\left(1-\frac xu\right)\right)} \, dx\\
& = u\cdot 
    \int_0^{\alpha \cdot \frac 1u} \frac{1-\frac{1}{8(1+\ln 2)\Delta} \cdot \left[ x \ge \frac12\right]}{\left(1- x\right)\left(1-\ln\left(1-x\right)\right)} \, dx\\
& = u\cdot 
    \int_0^{1-e^{1-e}} \frac{1-\frac{1}{8(1+\ln 2)\Delta} \cdot \left[ x \ge \frac12\right]}{\left(1- x\right)\left(1-\ln\left(1-x\right)\right)} \, dx\\
& = u\cdot \left(\int_0^{\frac12} \frac{1}{\left(1- x\right)\left(1-\ln\left(1-x\right)\right)} \, dx + \int_{\frac12}^{1-e^{1-e}} \frac{1-\frac{1}{8(1+\ln 2)\Delta}}{\left(1- x\right)\left(1-\ln\left(1-x\right)\right)} \, dx\right)\\
& = u \cdot \left( \ln(1+\ln 2) + \left(1-  
\frac{1}{8(1+\ln 2)\Delta}\right) \cdot \left(1 - \ln(1+\ln2)\right)\right)\\
&= u \cdot \left( 1 -\frac{1}{8(1+\ln 2)\Delta} \cdot \left(1 - \ln(1+\ln2)\right)\right)\\
&= u \cdot \left(1 - \frac{1- \ln(1+\ln2)}{8(1+\ln2)} \cdot \frac1\Delta\right)\\
&\le  u \cdot \left(1 - \frac{0.03}\Delta\right),
\end{align*}
where the first equality follows by a change of variable and the fourth equality follows from $\int\frac{1}{(1-x)(1-\ln(1-x))} \, dx = \ln(1-\ln(1-x))+c$. Thus, we have proved that $\E[S] \leq n\cdot u \cdot (1-\frac{0.03}{\Delta})$.

We are now ready to bound the expected number of matched ads. Let $A_k$ be the number of ads matched after sampling $k$ users (allowing also $k>u\cdot n$). Note $S\leq k \iff A_{\lfloor k\rfloor} \geq m$. Moreover, for $b \leq a+c$, $A_a+c \geq A_{b}$. Let $X = A_{\lceil n \cdot u \cdot (1 - \frac{0.03}{\Delta}) \rceil}$ be the number of ads matched with the first $\lceil n \cdot u \cdot (1 - \frac{0.03}{\Delta}) \rceil$ users. We have,
\begin{align*}
\Pr\left[X \geq m - \epsilon \cdot u \cdot n\left(1 - \frac
{0.03}{\Delta}\right)\right] 
& = \Pr\left[A_{\lceil n \cdot u \cdot (1 - \frac{0.03}{\Delta}) \rceil} + \epsilon u n\left(1- \frac{0.03}{\Delta}\right) \geq m\right] \\
& \geq \Pr\left[A_{\lfloor n \cdot u \cdot (1 - \frac{0.03}{\Delta}) + \epsilon u n(1- \frac{0.03}{\Delta}) \rfloor} \geq m\right]\\
& \geq \Pr[A_{\lfloor (1+\epsilon)\E[S] \rfloor} \geq m] \\
& = \Pr[S\leq (1+\epsilon)\E[S]]\\
& \geq 1 - \frac{1}{n}.
\end{align*}
Therefore, $\E[X] \geq u\cdot n \cdot (1-e^{1-e}) - O(\sqrt{n \ln (n)})$. 
Consider the next $\lceil nu\rceil - \lceil n \cdot u \cdot (1 - \frac{0.03}{\Delta}) \rceil \geq \lfloor n\cdot u \cdot \frac{0.03}{\Delta} \rfloor$ users. Each of these users has a probability of at least $\left(1 - \frac{X+\frac{0.03}{\Delta}\cdot u\cdot n}{n}\right)$ of being matched to an ad, let $Y$ be the total number of these users that get matched. Finally, we have,

\begin{align*}
\E[A_{\lceil u n\rceil}] & = \E[X]  +\E[Y]\\
& = \E[X] + \E[\E[Y\mid X]] \\
&\geq \E[X] +\E\left[n\cdot u\cdot \frac{0.03}{\Delta} \cdot \left(1 - \frac{X+\frac{0.03}{\Delta}\cdot u\cdot n}{n}\right) - 1\right]\\
& = \E[X] + \E\left[\left(n - X - \frac{0.03}{\Delta}\cdot u \cdot n\right) \frac{0.03}{\Delta}u\right] - 1\\
& = \E[X] + \left(1-\frac{0.03}{\Delta}u\right)\frac{0.03}{\Delta}un - \frac{0.03}{\Delta}\cdot u \cdot \E[X] - 1\\
&= \left(1-\frac{0.03}{\Delta} u\right)\E[X] + \left(1-\frac{0.03}{\Delta}u\right)\frac{0.03}{\Delta}un -1\\
&\geq \left(1-\frac{0.03}{\Delta}u\right)un(1-e^{1-e}) + \left(1-\frac{0.03}{\Delta}u\right)\frac{0.03}{\Delta}un - O\left(\sqrt{n\ln n}\right)\\
& = \left(1-\frac{0.03}{\Delta}u\right)\cdot \left(1 -e^{1-e} + \frac{0.03}{\Delta}\right) \cdot u \cdot n - O\left(\sqrt{n\ln n}\right)\\
&= \left( (1-e^{1-e}) + \frac{0.03}\Delta -\frac{0.03u}\Delta(1-e^{1-e})  - \frac{0.0009u}{\Delta^2}\right) \cdot u \cdot n - O\left(\sqrt{n\ln n}\right)\\
&\ge \left( (1-e^{1-e}) + \frac{0.03}\Delta -\frac{0.03}\Delta(1-e^{1-e})  - \frac{0.0009}{\Delta}\right) \cdot u \cdot n - O\left(\sqrt{n\ln n}\right)\\
&= \left( (1-e^{1-e}) + \frac{0.03 \cdot e^{1-e}}\Delta   - \frac{0.0009}{\Delta}\right) \cdot u \cdot n - O\left(\sqrt{n\ln n}\right)\\
&\ge \left( 1-e^{1-e} + \frac{0.004}\Delta \right) \cdot u \cdot n - O\left(\sqrt{n\ln n}\right)\\
&\ge \left( 1-e^{1-e} + \frac{0.004}{2\cdot \frac1{1-u}}\right) \cdot u \cdot n - O\left(\sqrt{n\ln n}\right)\\
&= \left( 1-e^{1-e} + 0.002 \cdot (1-u)\right) \cdot u \cdot n - O\left(\sqrt{n\ln n}\right),
\end{align*}
where in the last inequality we use $\Delta=\left\lceil \frac{1}{1-u}\right\rceil \leq 2\cdot \frac{1}{1-u}$ for $u\in(0,1)$.
\end{appendixproof}

We can now conclude that $u = 1$ gives the strongest bound of $1-\frac{e}{e^e}$:%
\begin{theoremrep}\label{thm:optimal_u}
The strongest impossibility result for the competitive ratio of online matching that can be obtained with the limiting constructions of \Cref{def:generalized-u-instance} is $1 - \frac{e}{e^e} \pm o(1)$. This bound can be obtained by setting $u = 1$, that is, with the construction of \Cref{def:instance}.   
\end{theoremrep}
\begin{appendixproof}
If $u \in (0,1)$, by \Cref{lem:online-matching-u-instance}, the competitive ratio with the distribution $D_u$ of \Cref{def:generalized-u-instance} is strictly larger --- that is, better --- than $1 - \frac e{e^e} + c \cdot (1-u)$, for a fixed constant $c > 0$. 
    
If $u = 1$, %
by \Cref{lem:analysis-online-u1-tight} and \Cref{thm:mm}, the competitive ratio with $u=1$ cannot be smaller than $1-e^{1-e}-o(1)$. Moreover, by \Cref{thm:mainlb}, the competitive ratio is no larger than $1- e^{1-e} + o(1)$. 
\end{appendixproof}

\section{Conclusion}\label{sec:conclusion}
We proved that the optimal competitive ratio of several online bipartite matching problems cannot be larger than $1-\frac{e}{e^e}$. The simplicity of our upper bound expression, and its similarity to the optimal competitive ratio for adversarial graphs that are also adversarially permuted, naturally raise the question of its potential optimality. We leave as open questions whether this upper bound is tight in (i) irregular cuckoo hashing instances, in (ii) IID instances with known distribution, in (iii) IID instances with unknown distribution, and in the hardest case, that of (iv) adversarial graphs with users permuted uniformly at random.

\section*{Acknowledgments}
We thank Will Ma, David Wajc, Pan Xu, and the anonymous reviewers for several comments and suggestions.%

\bibliographystyle{plain}
\bibliography{main}

\newpage
\appendix

\end{document}